\documentclass[12pt]{article}
\usepackage{latexsym,amsmath,amssymb,amsthm,natbib, commath, graphicx, enumerate}

\newcommand{\bb}{\mbox{\boldmath $b$}}
\newcommand{\bc}{\mbox{\boldmath $c$}}
\newcommand{\bd}{\mbox{\boldmath $d$}}
\newcommand{\be}{\mbox{\boldmath $e$}}

\newcommand{\bg}{\mbox{\boldmath $g$}}

\newcommand{\bu}{\mbox{\boldmath $u$}}
\newcommand{\bv}{\mbox{\boldmath $v$}}

\newcommand{\bx}{\mbox{\boldmath $x$}}
\newcommand{\by}{\mbox{\boldmath $y$}}

\newcommand{\bA}{\mbox{\boldmath $A$}}
\newcommand{\bB}{\mbox{\boldmath $B$}}
\newcommand{\bC}{\mbox{\boldmath $C$}}
\newcommand{\bD}{\mbox{\boldmath $D$}}
\newcommand{\bE}{\mbox{\boldmath $E$}}

\newcommand{\bG}{\mbox{\boldmath $G$}}
\newcommand{\bH}{\mbox{\boldmath $H$}}
\newcommand{\bI}{\mbox{\boldmath $I$}}
\newcommand{\bJ}{\mbox{\boldmath $J$}}
\newcommand{\bK}{\mbox{\boldmath $K$}}
\newcommand{\bL}{\mbox{\boldmath $L$}}
\newcommand{\bM}{\mbox{\boldmath $M$}}
\newcommand{\bN}{\mbox{\boldmath $N$}}

\newcommand{\bP}{\mbox{\boldmath $P$}}
\newcommand{\bQ}{\mbox{\boldmath $Q$}}
\newcommand{\bR}{\mbox{\boldmath $R$}}
\newcommand{\bS}{\mbox{\boldmath $S$}}
\newcommand{\bT}{\mbox{\boldmath $T$}}
\newcommand{\bU}{\mbox{\boldmath $U$}}
\newcommand{\bV}{\mbox{\boldmath $V$}}

\newcommand{\bX}{\mbox{\boldmath $X$}}
\newcommand{\bY}{\mbox{\boldmath $Y$}}

\newcommand{\bzero}{\mbox{\bf 0}}
\newcommand{\bone}{\mbox{\bf 1}}

\newcommand{\bbeta}{\mbox{\boldmath $\beta$}}

\newcommand{\bgamma}{\mbox{\boldmath $\gamma$}}

\newcommand{\bvarepsilon}{\mbox{\boldmath $\varepsilon$}}

\newcommand{\bmu}{\mbox{\boldmath $\mu$}}
\newcommand{\btheta}{\mbox{\boldmath $\theta$}}
\newcommand{\bTheta}{\mbox{\boldmath $\Theta$}}

\newcommand{\bLambda}{\mbox{\boldmath $\Lambda$}}

\newcommand{\bSigma}{\mbox{\boldmath $\Sigma$}}

\newcommand{\bphi}{\mbox{\boldmath $\phi$}}

\newcommand{\bxi}{\mbox{\boldmath $\xi$}}

\newcommand{\btau}{\mbox{\boldmath $\tau$}}

\newcommand{\argmin}{\operatornamewithlimits{arg\,min}}
\newcommand{\argmax}{\operatornamewithlimits{arg\,max}}
\newcommand{\mymap}{\operatorname*{\mapsto}}
\newcommand{\var}{\mbox{Var}}
\newcommand{\cov}{\mbox{Cov}}
\newcommand{\tr}{\mbox{tr}}

\newcommand{\el}{_{\ell}}

\newtheorem{lemma}{Lemma}
\newtheorem{thm}{Theorem}

\newtheorem{defn}{Definition}
\newcommand{\vecop}{\mbox{vec}}

\newcommand{\ratherbigpage}{\setlength{\textwidth}{6.5in}
  \setlength{\oddsidemargin}{-6pt}
  \setlength{\evensidemargin}{-6pt}
  \setlength{\textheight}{9.25in}
  \setlength{\topmargin}{0pt}
  \setlength{\headheight}{0pt}
  \setlength{\headsep}{5pt}}

\ratherbigpage

\begin{document}

\title{Varying-smoother models \\ for functional responses}
\author{Philip T.\ Reiss$^{1,2,*}$, Lei Huang$^3$, Huaihou Chen$^4$, and Stan Colcombe$^2$\bigskip\\$^1$Department of Child and Adolescent Psychiatry\\ and Department of Population Health, New York University
\\ $^2$Nathan S. Kline Institute for Psychiatric Research 
\\$^3$Department of Biostatistics, Johns Hopkins University\\$^4$Department of Biostatistics, University of Florida\\$^*$phil.reiss@nyumc.org
}

\maketitle
\begin{abstract}
This paper studies estimation of a smooth function $f(t,s)$ when we are given functional responses of the form $f(t,\cdot)+\mbox{error}$, but scientific interest centers on the collection of functions $f(\cdot,s)$ for different $s$. The motivation comes from studies of human brain development, in which $t$ denotes age whereas $s$ refers to brain locations. Analogously to varying-coefficient models, in which the mean response is linear in $t$, the ``varying-smoother'' models that we consider exhibit nonlinear dependence on $t$ that varies smoothly with $s$. We discuss three approaches to estimating varying-smoother models: (a) methods that employ a tensor product penalty; (b) an approach based on smoothed functional principal component scores; and (c) two-step methods consisting of an initial smooth with respect to $t$ at each $s$, followed by a postprocessing step. For the first approach, we derive an exact expression for a penalty proposed by Wood, and an adaptive penalty that allows smoothness to vary more flexibly with $s$. We also develop ``pointwise degrees of freedom,'' a new tool for studying the complexity of estimates of $f(\cdot,s)$ at each $s$. The three approaches to varying-smoother models are compared in simulations and with a diffusion tensor imaging data set.

\end{abstract}

\section{Introduction}\label{intro}
This article is concerned with functional responses that depend nonlinearly on a scalar predictor. 
The data for the $i$th of the $n$ given independent observations are assumed to be
\begin{equation}\label{thedata}t_i,\hspace{5mm}y_i(s_1),\ldots,y_i(s_L),\end{equation}
where $t_i$ lies in a domain ${\cal T}\subset {\mathbb R}$ and $s_1<\ldots<s_L$ is a fixed dense grid of points spanning a finite interval $\cal S\subset {\mathbb R}$; following \cite{ramsay2005}, we conceptualize this as having observed the entire function $y_i:\cal S\longrightarrow {\mathbb R}$.
We assume that these functional responses arise from the model
\begin{equation}\label{themod}y_i(s)= f(t_i, s) +  \varepsilon_i(s)\mbox{ for all }s\in{\cal S},\end{equation}
where $f$ is some smooth function on ${\cal T\times S}\subset{\mathbb R}^2$ 
 and  $\varepsilon_i$ is drawn from a zero-mean random error process on $\cal S$.
 
We wish to estimate $f$, and in particular we are interested in the family of functions $\{f(\cdot , s): s\in{\cal S}\}$.  
The motivation for this interest comes from studies of human development. When $t$ represents age, $f(\cdot , s)$ is the mean, as a function of age,  of the quantity measured by $y$ at point $s$. In the specific example that motivated this work, $\cal S$ is a set of locations in the brain, and $y(s)$ denotes fractional anisotropy (FA), a measure of white matter integrity, at location $s$. Thus $f(t,s)$ denotes the mean FA for that location at age $t$, and the function $f(\cdot , s)$ is what neuroscientists often refer to as the ``developmental trajectory'' of FA at location $s$. As a brief example of the scientific meaning of such trajectories, suppose that for given $s$, $f(t, s)$ characteristically increases with $t$ up to some point $t_s\equiv\argmax_t f(t,s)$, then decreases. Then the peak age $t_s$ can provide information about typical maturation for location $s$, and can be compared between diagnostic groups to study the links between psychiatric disorders and brain development \citep{shaw2007}.

In our data set, FA was measured at 107 voxels (volume units) in 146 individuals age 7--48. These $1\times 1\times 1$ mm voxels, based on registration to the FMRIB58\_FA standard space image (http://fsl.fmrib.ox.ac.uk/fsl/fslwiki/FMRIB58\_FA), trace a path along a midsagittal cross-section of the corpus callosum (see \figref{stan}). We take $s$ to represent arc length along this path, which ranges from $s_1=0$ mm (the leftmost point in the figure, toward the back of the brain) to $s_{107}=110.55$ mm. At right in \figref{stan}, a rainbow plot \citep{hyndman2010}, with FA curves color-coded by age, is used to visualize the relationship between age and the functional response. This relationship appears quite noisy and possibly non-monotonic (and hence nonlinear) in some locations. 
\begin{figure}
\begin{minipage}{\textwidth}
\centering
\includegraphics[width=.45\textwidth]{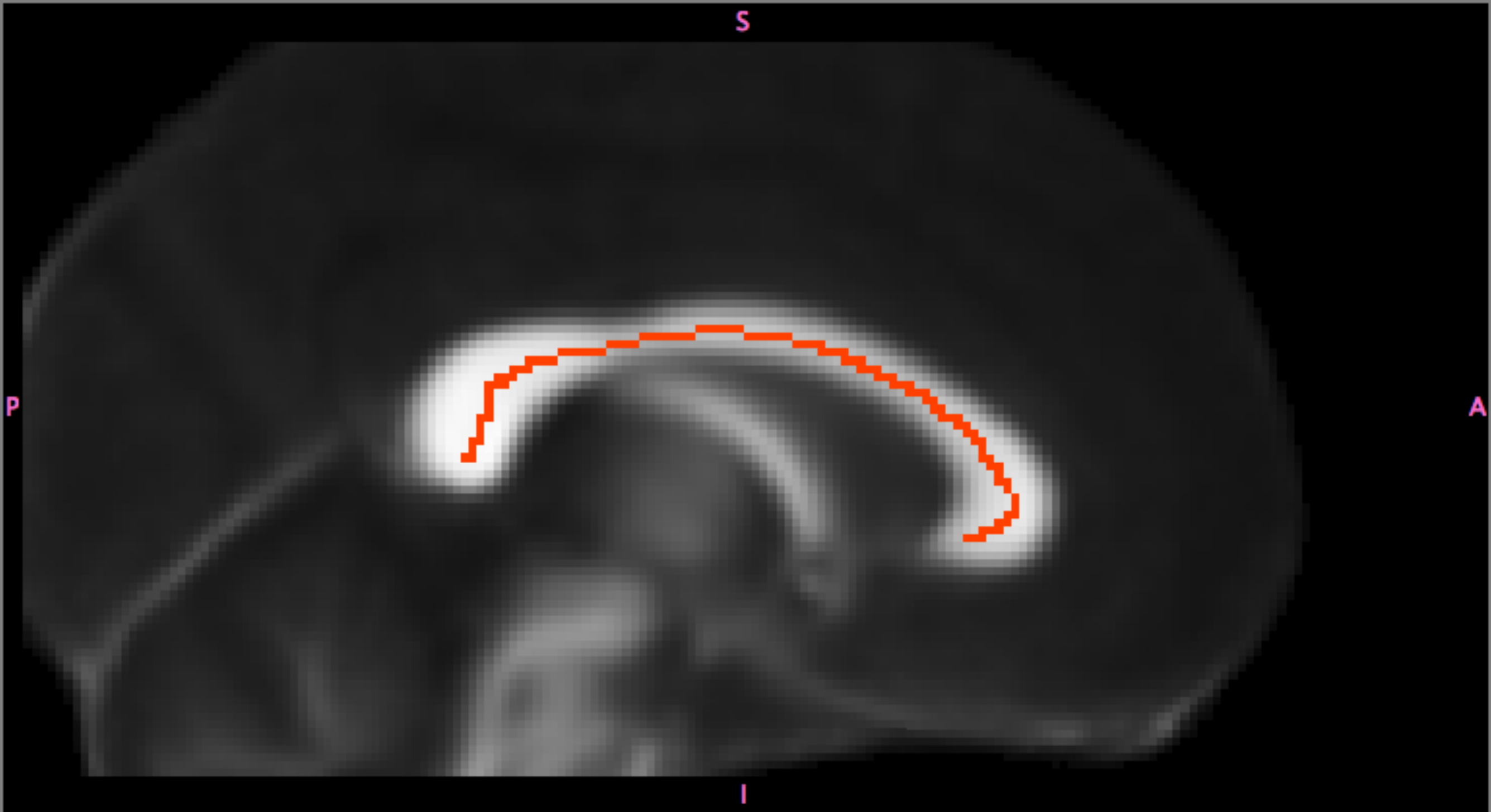}  
\includegraphics[width=.45\textwidth]{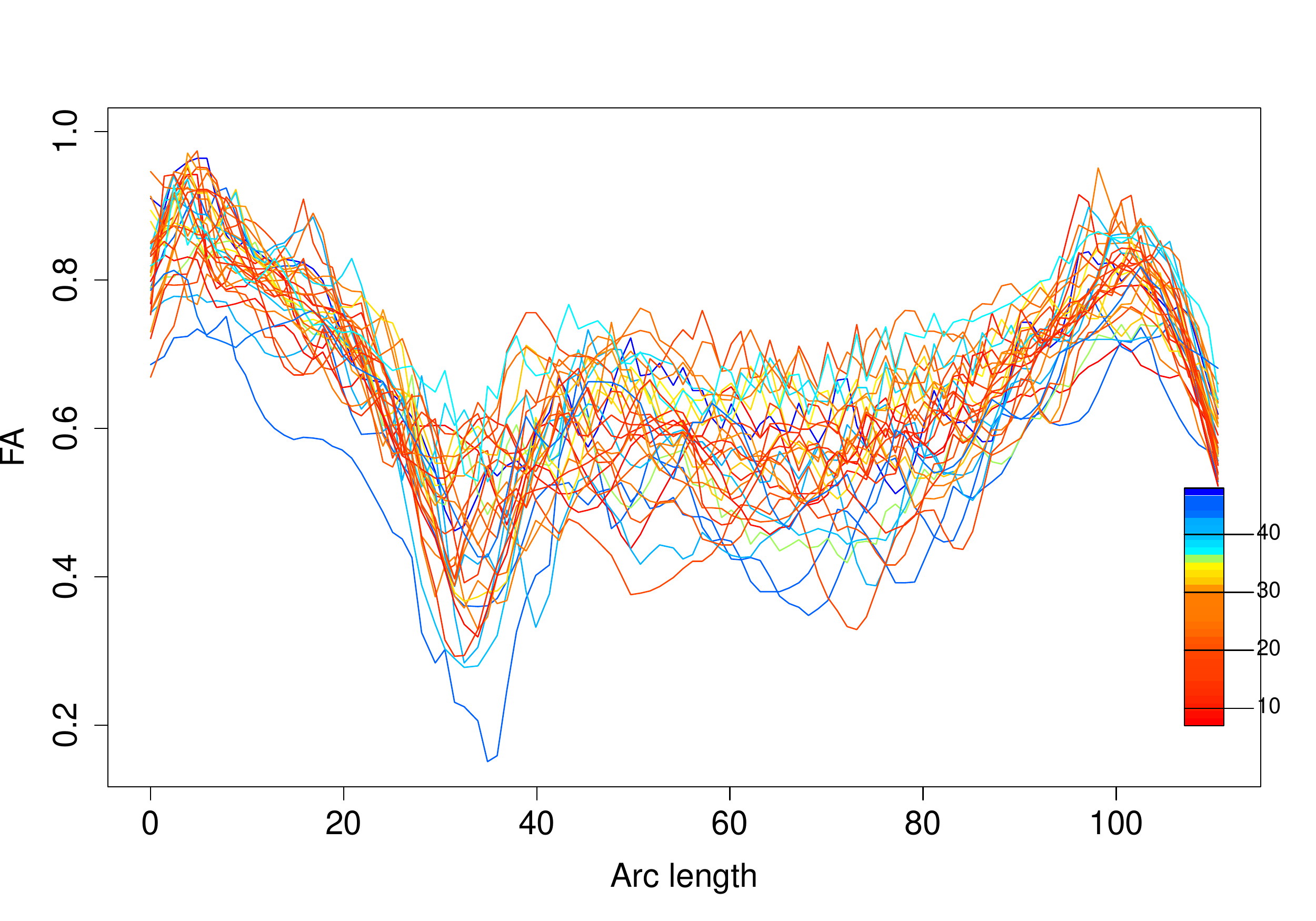}  
\label{stan}\caption{Left: Sequence of 107 corpus callosum voxels at which fractional anisotropy was recorded. Right: Rainbow plots displaying how the resulting FA profiles (functional responses) vary with age.}
\end{minipage}\end{figure}

Model \eqref{themod} for functional responses is not in itself new. Notably, \cite{greven2010} considered a more general model than \eqref{themod} for repeated functional responses. But we are aware of no previous treatments that have centered on the smooth functions $f(\cdot , s)$ 
and how they vary with $s$. To highlight this distinctive focus, we shall refer to \eqref{themod} as a \emph{varying-smoother} model. This term recalls the idea of a ``varying-coefficient'' model \citep{hastie1993}, which is what \eqref{themod} reduces to in the special case  $f(t_i, s)=t_i\beta(s)$. Varying-coefficient models are ordinarily defined for scalar (as opposed to functional) responses; specialized methodology is needed to estimate the varying coefficient $\beta(\cdot)$ in the functional-response case \citep{ramsay2005,reiss2010}. A similar point can be made regarding varying-smoother models. One can conceive of scalar-response applications in which one would like to estimate $\{f(\cdot , s): s\in{\cal S}\}$ on the basis of data $(y_i,t_i,s_i)$, $i=1,\ldots,n$, that are assumed to follow the model $y_i=f(t_i,s_i)+\varepsilon_i$. But the applications motivating our work involve functional responses, as in \eqref{thedata} and \eqref{themod}, and we shall restrict consideration to this setting.

To restate succinctly the basic distinction between the varying-coefficient and varying-smoother assumptions for model \eqref{themod}: in both cases $f(\cdot,s)$ varies in a smooth manner with $s$, but in the varying-coefficient case $f(\cdot,s)$ is linear for each $s$, whereas for varying-smoother models $f(\cdot,s)$ is in general nonlinear.

\cite{zhu2010, zhu2011} developed methodology for varying-coefficient models in which FA curves, similar to those considered here, depend linearly on age and other predictors. Their functional linear models were applied to an infant data set, whereas the much wider age range of our sample motivated our development of varying-smoother models to map the \emph{nonlinear} dependence of FA on age, along the corpus callosum, over a large portion of the lifespan.

To avoid possible confusion, we remark that varying-smoother models are very much distinct from fitting curves with ``varying smoothness.'' The latter refers, in the univariate case, to estimating a function $f(t)$ where the smoothness of $f$ varies with $t$---a goal often pursued by means of wavelets \citep{ogden1997} or by extensions of spline methodology \citep{krivobokova2008,storlie2010}. Our goal, by contrast, is to estimate $f(\cdot,s)$, a smooth function of $t$ that varies (smoothly) with $s$. 

Successful pursuit of this goal requires that we borrow strength across locations to a sufficient extent so that $f(\cdot,s)$ is more accurately estimated for each $s$, while still allowing the shape of $f(\cdot,s)$ to vary flexibly with $s$---since understanding this variation may be the principal scientific objective, for example in neurodevelopmental studies. Fully Bayesian modeling with spatially informed priors  \citep[e.g.,][]{fahrmeir2004,congdon2006}  might be a natural approach to this problem. However, in view of the high dimensionality of the functional responses in many applications, this approach may prove computationally prohibitive. The approaches of this paper rely on the roughness penalty paradigm that has been employed fruitfully in smoothing problems \citep{green1994,ruppert2003,wood2006} and functional data analysis \citep{ramsay2005}.

To help readers through what will be a rather algebra-heavy presentation,the next section collects the main notations used below, as well as stating our key assumptions. Sections~\ref{tpp} through \ref{2s} describe three basic approaches to fitting varying-smoother models. Section~\ref{pwdf} introduces \emph{pointwise degrees of freedom}, a novel tool for assessing and comparing the model complexity (with respect to $t$) of different estimates of $f(\cdot,s)$. The different approaches to varying-smoother modeling are compared in a simulation study in Section~\ref{simsec}, and applied to the corpus callosum FA data in Section~\ref{realsec}. Concluding remarks are offered in Section~\ref{discsec}.

\section{Notation and assumptions}\label{setup}
 In most of what follows we consider only a single real-valued predictor $t_i$ ($i=1,\ldots, n$). The $i$th response is a function $y_i(\cdot)$ observed at a common set of points $s_1,\ldots,s_L$, giving rise to an $n\times L$ response matrix 
\[\bY=\left(\begin{array}{ccc}y_{11} &\ldots& y_{1L} \\ \vdots & \ddots & \vdots \\ y_{n1} &\ldots& y_{nL}\end{array}\right)=\left(\begin{array}{ccc}y_1(s_1) &\ldots& y_1(s_L) \\ \vdots & \ddots & \vdots \\ y_n(s_1) &\ldots& y_n(s_L)\end{array}\right).\]  
Our methods can be extended to irregularly sampled functions by adding a presmoothing step \citep[cf.][]{chiou2003}. Let $\by_{i\cdot}^T$ and $\by_{\cdot \ell}$ denote the $i$th row and $\ell$th column of $\bY$, respectively, and let $\by=\vecop(\bY)=\left(\begin{array}{c}\by_{\cdot 1}\\\vdots\\ \by_{\cdot L}\end{array}\right)\in{\mathbb R}^{nL}$.  

Analogous notation ($\hat{\bY},\hat{\by}$, etc.)\ will be used for fitted values from our procedures for fitting model \eqref{themod}, to be described in Sections~\ref{tpp} through \ref{2s}. For all of these procedures, the fitted values can be written as $\hat{\by}=\mathbf{\cal H}\by$ for some $nL\times nL$ ``hat" matrix 
\begin{equation}\label{blockhat}\mathbf{\cal H}=\left(\begin{array}{ccc}\mathbf{\cal H}_{11} & \ldots & \mathbf{\cal H}_{1L} \\ \vdots & \ddots & \vdots \\ \mathbf{\cal H}_{L1} & \ldots & \mathbf{\cal H}_{LL}\end{array}\right),\end{equation} where each of the blocks $\mathbf{\cal H}_{\ell_1\ell_2}$ is $n\times n$; we shall denote the $(i,j)$ entry of the $(\ell_1,\ell_2)$ block by $h_{(\ell_1\ell_2)ij}$.

  We take the domain of $f(\cdot,\cdot)$ to be ${\cal T}\times{\cal S}$, where both the ``temporal'' domain $\cal T$ and the ``spatial'' or functional-response domain $\cal S$ are finite intervals on the real line. The key building blocks for our estimators of $f$ will be a basis of $K_t\leq n$ smooth functions, such as $B$-splines, defined on $\cal T$; another set of $K_s\leq L$ basis functions defined on $\cal S$; and associated penalty matrices $\bP_t$ and $\bP_s$,    respectively.  Let $\bb_t(t)=[b_{x1}(t),\ldots,b_{xK_t}(t)]^T$ where $b_{x1},\ldots,b_{xK_t}$ are the predictor-domain basis functions, and let  $\bb_s(s)=[b_{s1}(s),\ldots,b_{sK_s}(s)]^T$ where $b_{s1},\ldots,b_{sK_s}$ are the function-domain basis functions.  Define
\[\underbrace{\bB_t}_{(n\times K_t)}=\left[\begin{array}{c}\bb_t(t_1)^T\\\vdots\\\bb_t(t_n)^T\end{array}\right]\mbox{ and }\underbrace{\bB_s}_{(L\times K_s)}=\left[\begin{array}{c}\bb_s(s_1)^T\\\vdots\\\bb_s(s_L)^T\end{array}\right].\]
These two matrices are assumed to be of full rank. 

The temporal penalty matrix $\bP_t$ is a symmetric positive semidefinite $K_t\times K_t$ matrix such that, for a given function
\begin{equation}\label{betabx}g(t)=\bgamma^T\bb_t(t),\end{equation} 
we have $\bgamma^T\bP_t\bgamma=r_t(g)$ where $r_t(g)$ is some measure of the roughness of $g$. We ordinarily use the second-derivative penalty matrix
$\bP_t=[\int b^{\prime\prime}_{ti}(t)b^{\prime\prime}_{tj}(t)dt]_{1\leq i,j\leq K_t}$, for which  $\bgamma^T\bP_t\bgamma=r_t(g)\equiv\int_{\cal T}g^{\prime\prime}(t)^2dt$. 
 Difference penalties \citep{eilers1996}, another popular choice for $B$-spline bases, are somewhat simpler computationally, albeit without a corresponding to a closed-form functional $g\mapsto r_t(g)$.  Analogously, the $K_s\times K_s$  penalty matrix $\bP_s$ is associated with a spatial roughness index $r_s(\cdot)$. Below we shall also require the matrices $\bQ_t=[\int b_{ti}(t)b_{tj}(t)dt]_{1\leq i,j\leq K_t}$ and $\bQ_s=[\int b_{si}(s)b_{sj}(s)ds]_{1\leq i,j\leq K_s}$.

The tensor product of the two bases is the set of functions on ${\cal T}\times{\cal S}$ given by $\{(t,s) \mapsto b_{xi}(t)b_{sj}(s): 1\leq i\leq K_t, 1\leq j \leq K_s\}$.  The span of the tensor product basis comprises all functions of the form
\begin{equation}\label{tpdef}f(t,s)=\sum_{i=1}^{K_t}\sum_{j=1}^{K_s}\theta_{ij}b_{xi}(t)b_{sj}(s)=\bb_t(t)^T\bTheta\bb_s(s)\end{equation} for real-valued coefficients $\theta_{ij}$, where $\bTheta=(\theta_{ij})_{1\leq i\leq K_t, 1\leq j \leq K_s}$. For $f$ of this form, the observed data can be expressed in terms of the matrix equation
\begin{equation}\label{mateq}\bY=\bB_t\bTheta\bB_s^T+\bE,\end{equation}
where $\bE=[\varepsilon_i(s\el)]_{1\leq i \leq n, 1\leq \ell\leq L}$. Letting $\btheta=\vecop(\bTheta)$ and $\bvarepsilon=\vecop(\bE)$, \eqref{mateq} can be written in vector form as
$\by=(\bB_s\otimes\bB_t)\btheta+\bvarepsilon$. 

We shall require two assumptions regarding the spatial smoother: 
\begin{equation} \label{as1}\bB_s\bone_{K_s}=\bone_L.\end{equation}
\begin{equation}\label{as2}\bP_s\bone_{K_s}=\bzero_{K_s}.\end{equation}
These are mild assumptions, inasmuch as \eqref{as1} holds for a $B$-spline basis in one dimension, while \eqref{as2} holds for a derivative or difference penalty. In the sequel we refer repeatedly to ``splines,'' but our development encompasses any penalized basis functions for which \eqref{as1} and \eqref{as2} hold. Finally, let $\bJ_n=\bone_n\bone_n^T/n$.

\section{Tensor product penalty methods}\label{tpp}
In this and the next two sections we present three basic approaches to fitting the varying-smoother model $f$ by estimating $\bTheta$ in equation \eqref{mateq}. 

\subsection{Penalized OLS and penalized GLS}\label{polspgls}
The first approach is to solve \eqref{mateq} directly by penalized bivariate smoothing. A na\"{i}ve estimate of $\bTheta$ is
\begin{equation}\label{pols}\hat{\bTheta}=\argmin_{\Theta}\left[\|\bY-\bB_t\bTheta\bB_s^T\|_F^2+p(\bTheta)\right],\end{equation}
where $\|\cdot\|_F$ denotes the Frobenius norm $\|\bA\|^2_F=\tr(\bA^T\bA)$ and $p(\bTheta)$ is a bivariate roughness penalty, i.e., some nonnegative functional whose value increases with the roughness or wiggliness of the function $(t,s)\mapsto\bb_t(t)^T\bTheta\bb_s(s)$. Inclusion of $p(\bTheta)$ in the objective function serves to prevent overfitting. 

Since in most cases 
\begin{equation}\label{wfd}\cov[y(s_1),\ldots,y(s_L)|x]=\bSigma\neq\bI_L,\end{equation} it may be preferable to use a penalized generalized least squares (GLS) estimate 
\begin{equation}\label{pgls}\hat{\bTheta}=\argmin_{\Theta}\left[\left\|(\bY-\bB_t\bTheta\bB_s^T)\hat{\bSigma}^{-1/2}\right\|_F^2+p(\bTheta)\right],\end{equation}
for some precision (inverse covariance) matrix estimate $\hat{\bSigma}^{-1}$, rather than the penalized ordinary least squares estimate \eqref{pols}. (Since the covariance must be estimated, \eqref{pgls} is more correctly a penalized \emph{feasible} GLS estimate \citep{freedman2009}.)

To implement the penalized OLS estimate \eqref{pols} and the penalized GLS estimate \eqref{pgls}, we must attend to three details: (i) the form of the roughness penalty $p(\bTheta)$,  (ii) estimation of the precision matrix $\bSigma^{-1}$, and (iii) selection of the tuning parameters in $p(\bTheta)$, which govern the smoothness of the function estimate $\hat{f}(t,s)=\bb_t(t)^T\hat{\bTheta}\bb_s(s)$. The first of these issues is taken up in following subsection.  See Appendix~\ref{cbp} regarding precision matrix estimation. For smoothing parameter selection we use restricted maximum likelihood \citep[REML;][]{ruppert2003}; see Appendix~\ref{spsapp} for discussion of this topic. Section~\ref{vds} presents a new adaptive penalty (based on the penalty that we propose in Section~\ref{eew}) that allows greater flexibility in accommodating the varying smoothness of $f(\cdot,s)$ for different $s$. 

\subsection{Exact evaluation of Wood's tensor product penalty}\label{eew}
Although penalized smoothing with tensor product bases is not at all new, the form of $p(\bTheta)$ is still not a settled matter \citep{xiao2013}. 
Tensor product penalization usually builds upon given roughness functionals $r_t$ and $r_s$ for functions of $t$ and $s$ respectively.
 Here we adopt the proposal of \cite{wood2006tensor} to define a tensor product penalty as 
\begin{equation}\label{penf}\mbox{pen}(f)= \lambda_s\int_{\cal T} r_s[f(t,\cdot)]dt + \lambda_t\int_{\cal S} r_t[f(\cdot,s)]ds.\end{equation}
\cite{wood2006tensor} proposes an approximate procedure for computing these integrals, but exact evaluation is possible in our case, as shown by the following result.
 
\begin{thm} \label{tpen} For bivariate functions $f$ of form \eqref{tpdef}, penalty \eqref{penf} can be expressed as
 \begin{equation}\label{wct}\mbox{pen}(f)=\btheta^T\left[\lambda_s(\bP_s\otimes\bQ_t)+\lambda_t(\bQ_s\otimes\bP_t)\right]\btheta.\end{equation}\end{thm}
 
 The proof is given in Appendix~\ref{pt1}. Note that since splines are piecewise polynomials, the integrals defining $\bQ_t,\bQ_s$, as well as $\bP_t,\bP_s$ for derivative penalties, can be evaluated exactly by Newton-Cotes quadrature \citep[e.g.,][]{ralston2001}. Taking $p(\bTheta)$ to be penalty \eqref{wct}, we can express the penalized OLS estimate \eqref{pols} as the vector 
 \begin{equation}\label{grr}\hat{\btheta}=\argmin_{\theta}\left[\|\by-(\bB_s\otimes\bB_t)\btheta\|^2+\btheta^T\{\lambda_s(\bP_s\otimes\bQ_t)+\lambda_t(\bQ_s\otimes\bP_t)\}\btheta\right],\end{equation}
and the penalized GLS estimate \eqref{pgls} as
\begin{eqnarray}\hat{\btheta}&=&\argmin_{\theta}\left[\{\by-(\bB_s\otimes\bB_t)\btheta\}^T(\hat{\bSigma}^{-1}\otimes \bI_n)\{\by-(\bB_s\otimes\bB_t)\btheta\}\right.\nonumber\\&&\qquad\qquad\left.+\btheta^T\{\lambda_s(\bP_s\otimes\bQ_t)+\lambda_t(\bQ_s\otimes\bP_t)\}\btheta\right].\label{pglsvec}\end{eqnarray} 
We remark that, if both the $t$- and the $s$-basis are orthonormal, \eqref{wct} reduces to the penalty
 $\btheta^T\left[\lambda_s(\bP_s\otimes\bI_{K_t})+\lambda_t(\bI_{K_s}\otimes\bP_t)\right]\btheta$
used, for example, by \cite{eilers2003} and \cite{currie2006}.

\subsection{Adaptive (spatially varying) temporal smoothing}\label{vds}
A generalization of the temporal penalty $\lambda_t\int_{\cal S} r_t[f(\cdot,s)]ds$ is 
\begin{equation}\label{vdpen}\int_{\cal S} \lambda_t(s)r_t[f(\cdot,s)]ds,\end{equation}
 i.e., allowing the temporal smoothing parameter $\lambda_t$ to vary with $s$. This is a natural way to enable our estimates to adapt to varying smoothness of $f(\cdot, s)$, with respect to $t$, for different $s$. A relatively straightforward way to incorporate such a smoothly varying $\lambda_t$ is to assume that
\begin{equation}\label{smallbasis}\lambda_t(s)=\sum_{k=1}^{K_s^*}\lambda_{t,k} b^*_k(s)\end{equation}
for some $\lambda_{t,1},\ldots,\lambda_{t,K_s^*}\geq 0$, where $b^*_1,\ldots,b^*_{K_s^*}$ form a coarse $B$-spline basis on domain $\cal S$.
 Penalty \eqref{vdpen} then becomes $\sum_{k=1}^{K_s^*}\lambda_{t,k}\int_{\cal S} b^*_k(s)r_t[f(\cdot,s)]ds$, giving the modified tensor product penalty
\begin{equation}\label{modpen} \lambda_s\int_{\cal T} r_s[f(t,\cdot)]dt+\sum_{k=1}^{K_s^*}\lambda_{t,k}\int_{\cal S} b^*_k(s)r_t[f(\cdot,s)]ds \end{equation} [cf.\ \eqref{penf}].
This penalty is expressed as a quadratic form in the following result, which is proved in Appendix~\ref{pt1}. 

\begin{thm} \label{vdtpen} For bivariate functions $f$ of form \eqref{tpdef}, penalty \eqref{modpen} equals
\begin{equation}\label{ncpen}\btheta^T\left[\lambda_s(\bP_s\otimes\bQ_t)+\sum_{k=1}^{K_s^*}\lambda_{t,k}(\bQ^{b^*_k}_s\otimes\bP_t)\right]\btheta.\end{equation}
where $\bQ^{b^*_k}_s=[\int b^*_k(s)b_{si}(s)b_{sj}(s)ds]_{1\leq i,j\leq K_s}$. \end{thm}

Theorem~\ref{vdtpen} shows that we can let the smoothness of $f(\cdot,s)$ vary more flexibly with $s$ by solving another quadratically penalized least squares problem, but now with $K_s^*+1$ smoothing parameters instead of 2. Note that, since $\sum_{k=1}^{K_s^*}b^*_k(s)=1$ for all $s$, we have $\sum_{k=1}^{K_s^*}\bQ^{b^*_k}_s=\bQ_s$. Thus if $\lambda_{t,1}=\ldots=\lambda_{t,K_s^*}=\lambda_t$ then \eqref{ncpen} reverts to the constant-smoothing-parameter penalty \eqref{wct}.

\section{A smoothed functional principal component scores method}\label{fpc}
We consider next adapting a method of \cite{chiou2003} to the problem of varying-smoother modeling. These authors propose a single-index approach to modeling smooth dependence of functional responses on a set of scalar predictors. They assume that the functional responses arise from a stochastic process with a finite-dimensional Karhunen-Lo\`{e}ve or functional principal component (FPC) expansion, i.e., the $i$th response can be expressed uniquely as 
\begin{equation}\label{kle}y_i(s)=\mu(s)+\sum_{a=1}^A c_{ia}\phi_a(s)\end{equation} where $\mu(\cdot)$ is the mean function, $\phi_1(\cdot),\ldots,\phi_A(\cdot)$ are the leading principal component functions and  $c_{i1},\ldots,c_{iA}$ are the corresponding scores. In the present paper we are considering a single scalar predictor $t$, for which the proposed model of \cite{chiou2003} reduces to 
\begin{equation}\label{chioumod}E[y(s)|x]=\mu(s)+\sum_{a=1}^A g_a(t)\phi_a(s),\end{equation}
for some smooth functions $g_1,\ldots,g_A:{\cal T}\longrightarrow\mathbb{R}$. These $A$ functions can be estimated separately by smoothing the corresponding estimated FPC scores. Thus we fit model \eqref{chioumod} in two steps:
\begin{enumerate}
\item Derive estimates $\hat{\mu}(\cdot),\hat{c}_{ia},\hat{\phi}_a(\cdot)$ ($i=1,\ldots,n,a=1,\ldots,A$) of the unknowns in \eqref{kle}.
\item For $a=1,\ldots,A$, apply nonparametric regression to the ``data'' $(t_1,\hat{c}_{1a}),\ldots, (t_n,\hat{c}_{na})$ to obtain an estimate $\hat{g}_a$ of $g_a$.
\end{enumerate}
\cite{chiou2003} estimate the model by local linear smoothing, but note that splines can be used as well. In Appendix~\ref{compchiou} we outline a penalized spline implementation that produces an estimate of the coefficient matrix $\bTheta$ in \eqref{mateq}.

\section{Two-step methods}\label{2s}
Modeling approaches of the third and final type that we consider proceed by (1) obtaining an initial estimate $\tilde{f}\el$ of $f(\cdot,s\el)$, separately for each $\ell=1,\ldots,L$; and (2) a ``postprocessing'' step that combines these function estimates into an estimate of $f$, via smoothing and/or projection. A similar two-step scheme was developed by \cite{fan2000} for varying-coefficient models with functional responses. We discuss each of the two steps in turn.

\subsection{Step 1}
In the first step we obtain, for $\ell=1,\ldots,L$, a standard penalized spline estimate $\tilde{f}\el(\cdot)=\tilde{\bxi}\el^T\bb_t(\cdot)$ where 
\begin{eqnarray*}\tilde{\bxi}\el&=&\argmin_{\xi\in{{\mathbb R}^{K_t}}}\left(\|\by_{\cdot \ell}-\bB_t\bxi\|^2+\lambda_{t \ell}\bxi^T\bP_t\bxi\right) \\
&=&(\bB_t^T\bB_t+\lambda_{t \ell}\bP_t)^{-1}\bB_t^T \by_{\cdot\ell}.\end{eqnarray*}
The smoothing parameter $\lambda_{t \ell}$ is allowed to vary with $\ell$, to adapt to varying smoothness of $f(\cdot,s)$ for different $s$. \cite{reiss2014} propose a fast algorithm for choosing the optimal $\lambda_{t \ell}$, in the sense of the REML criterion, for each $\ell=1,\ldots,L$ with large $L$. They also derive a useful expression for the fitted value matrix $\tilde{\bY}=[\tilde{f}\el(t_i)]_{1\leq i\leq n,1\leq\ell\leq L}$ by means of Demmler-Reinsch orthogonalization, as follows.
First find a $K_t\times K_t$ matrix $\bR_t$ such that $\bR_t^T\bR_t=\bB_t^T\bB_t$, e.g., by Cholesky decomposition.   Define $\bU_t\mbox{Diag}(\btau)\bU_t^T$, where $\btau=(\tau_1,\ldots,\tau_{K_t})^T$, as the singular value decomposition of $\bR_t^{-T}\bP_t\bR_t^{-1}$.  We then have  
\begin{equation}\tilde{\bY}=\bA_t[\bM\odot(\bA_t^T\bY)]\label{ytilde}\end{equation}
where $\odot$ denotes Hadamard (componentwise) product, 
\begin{eqnarray}\label{defm}\bM&=&\left(\frac{1}{1+\lambda_{t \ell} \tau_k}\right)_{1\leq k\leq K_t, 1\leq \ell\leq L},\mbox{ and}
\\ \label{defa}\bA_t&=&\bB_t\bR_t^{-1}\bU_t.\end{eqnarray}

\subsection{Step 2}
We consider three variants of step 2, in which we refine the initial set of pointwise smoothers.

\subsubsection{Penalized variant}\label{penvt}
The simplest variant is to apply a spatial smoother, given by some $L\times L$ matrix $\bH_s$, to each of the rows $\tilde{\by}_{1\cdot},\ldots,\tilde{\by}_{n\cdot}$ of the initial fitted value matrix $\tilde{\bY}$. By \eqref{ytilde}, this results in the final fitted values
\begin{equation}\label{fit2}\hat{\bY}=\tilde{\bY}\bH_s^T=\bA_t[\bM\odot(\bA_t^T\bY)]\bH_s^T.\end{equation}
  In particular, using the standard penalized basis smoother $\bH_s=\bB_s(\bB_s^T\bB_s+\lambda_s\bP_s)^{-1}\bB_s^T$, \eqref{fit2} implies the function estimate
\begin{equation}\label{gest}\hat{f}(t,s)=\bb_t(t)^T\bR_t^{-1}\bU_t[\bM\odot(\bA_t^T\bY)]\bB_s(\bB_s^T\bB_s+\lambda_s\bP_s)^{-1}\bb_s(s),\end{equation}
which has the tensor product form (\ref{tpdef}).

\subsubsection{FPC variant}\label{fpcvt}
  A possible disadvantage of performing step~2 by simple spatial smoothing is that it fails to take advantage of the patterns of variation in the responses as revealed by functional PCA.
  An alternative for step~2 is to project $\tilde{\by}_{1\cdot}-\hat{\bmu},\ldots,\tilde{\by}_{n\cdot}-\hat{\bmu}$ onto the span of $\hat{\bphi}_1,\ldots,\hat{\bphi}_A$ for some $A$, where $\hat{\bmu},\hat{\bphi}_a\in\mathbb{R}^L$ are discretized estimates of the mean function and the $a$th FPC function, respectively. Appendix~\ref{fpcvs} provides the details.

\subsubsection{Penalized FPC variant}\label{penfpcvt}
For $i=1,\ldots,n$, the penalized variant of step~2 refines the initial fitted values $\tilde{f}_1(t_i),\ldots,\tilde{f}_L(t_i)$ by applying a penalized smoother to them. The FPC variant, on the other hand, projects these values (after centering them) onto the span of the leading FPC functions. As shown in Appendix~\ref{fpcvs}, it is straightforward to combine these two approaches to postprocessing. \cite{reiss2007} found that a similar hybrid of FPC expansion and roughness penalization worked well for regressing scalar responses on functional predictors.

\section{Pointwise degrees of freedom}\label{pwdf}
Varying-smoother models seek to estimate the smooth bivariate function $f$ while allowing for differing smoothness or complexity of the function $f(\cdot,s)$ for different $s$. In this section we introduce a notion of pointwise degrees of freedom that quantifies the model complexity of an estimate of $f(\cdot,s)$. 

\subsection{Definition}
Consider first the matrix $\mathbf{\cal H}$ such that $\tilde{\by}=\left(\begin{array}{c}\tilde{\by}_{\cdot 1}\\\vdots\\\tilde{\by}_{\cdot L}\end{array}\right)\equiv\vecop(\tilde{\bY})$, the concatenation of the $L$ separate smooths produced in step 1 of the two-step method, is given by $\tilde{\by}=\mathbf{\cal H}\by$.  Referring to the block form \eqref{blockhat} of the hat matrix, we have, for $\ell=1,\ldots,L$:
 \begin{enumerate}[(a)]
 \item  $\mathbf{\cal H}_{\ell\ell}=\bB_t(\bB_t^T\bB_t+\lambda_{t \ell}\bP_t)^{-1}\bB_t^T$;
 \item $\mathbf{\cal H}_{\ell\ell^*}=0$ for each $\ell^*\neq \ell$;
 \item $\tilde{\by}_{\cdot \ell}=\mathbf{\cal H}_{\ell\ell}\by_{\cdot \ell}$.
 \end{enumerate}
In this case there is no need for a novel definition of pointwise degrees of freedom: the df 
of the $\ell$th-location model can be defined in the conventional manner \citep{buja1989}, as 
\begin{equation}\label{odf}\tr(\mathbf{\cal H}_{\ell\ell})=\sum_{i=1}^n h_{(\ell\ell)ii}=\sum_{i=1}^n\frac{\partial \tilde{y}_{i\ell}}{\partial y_{i\ell}}.\end{equation} 

 On the other hand, when a smoothing procedure shares information across locations, the fitted values $\hat{y}_{i\ell}$ depend on the entire functional response datum $\by_{i\cdot}$ rather than solely on its $\ell$th component $y_{i\ell}$. The standard definition of df is then inadequate.  We therefore propose the following generalization. 
\begin{defn}\label{pwdfdef} The \emph{pointwise effective degrees of freedom} at location $\ell$ is 
\begin{equation}\label{dff}d\el=\sum_{i=1}^n\sum_{\ell^*=1}^L\frac{\partial \hat{y}_{i\ell}}{\partial y_{i\ell^*}}=\sum_{\ell^*=1}^L\tr(\mathbf{\cal H}_{\ell\ell^*}).\end{equation}
We shall use $\bd=(d_1,\ldots,d_L)^T$ as a generic notation for the vector of pointwise df values obtained by any of the methods discussed below.\end{defn}
   The above definition implies
\begin{equation}\label{dfh}d\el=\tr\left[(\be\el^T\otimes\bI_n)\mathbf{\cal H}(\bone_L\otimes\bI_n)\right],\end{equation}
where $\be\el$ is the $L$-dimensional vector with 1 in the $\ell$th position and 0 elsewhere.

Some intuition for Definition~\ref{pwdfdef} can be gained from \figref{hatfig}. Subfigure~(a) displays a portion of the block hat matrix \eqref{blockhat} obtained by fitting a tensor product penalty smooth (as in Section~\ref{tpp}) to a subset of the corpus callosum data. Had we fitted separate models at each voxel, the nonzero entries in the hat matrix---representing influence of the responses on the fitted values---would be confined to diagonal blocks such as those outlined in black. The sharing of information across locations is expressed as a ``blockwise blurring'' in the horizontal direction, which serves as the motivation for Definition~\ref{pwdfdef}. Consider a toy example with $n=5$ observations and $L=4$ locations, so that the hat matrix \eqref{blockhat} comprises a $4\times 4$ grid of $5\times 5$ blocks. If separate models are fitted at each location, the ordinary df for the 2nd-location model is the sum of the shaded values in \figref{hatfig}(b). The proposed pointwise df for the 2nd location, which takes into account the influence of neighboring locations in a varying-smoother model, is the sum of the shaded values in  \figref{hatfig}(c).
\begin{figure}\centering
\includegraphics[width=.85\textwidth]{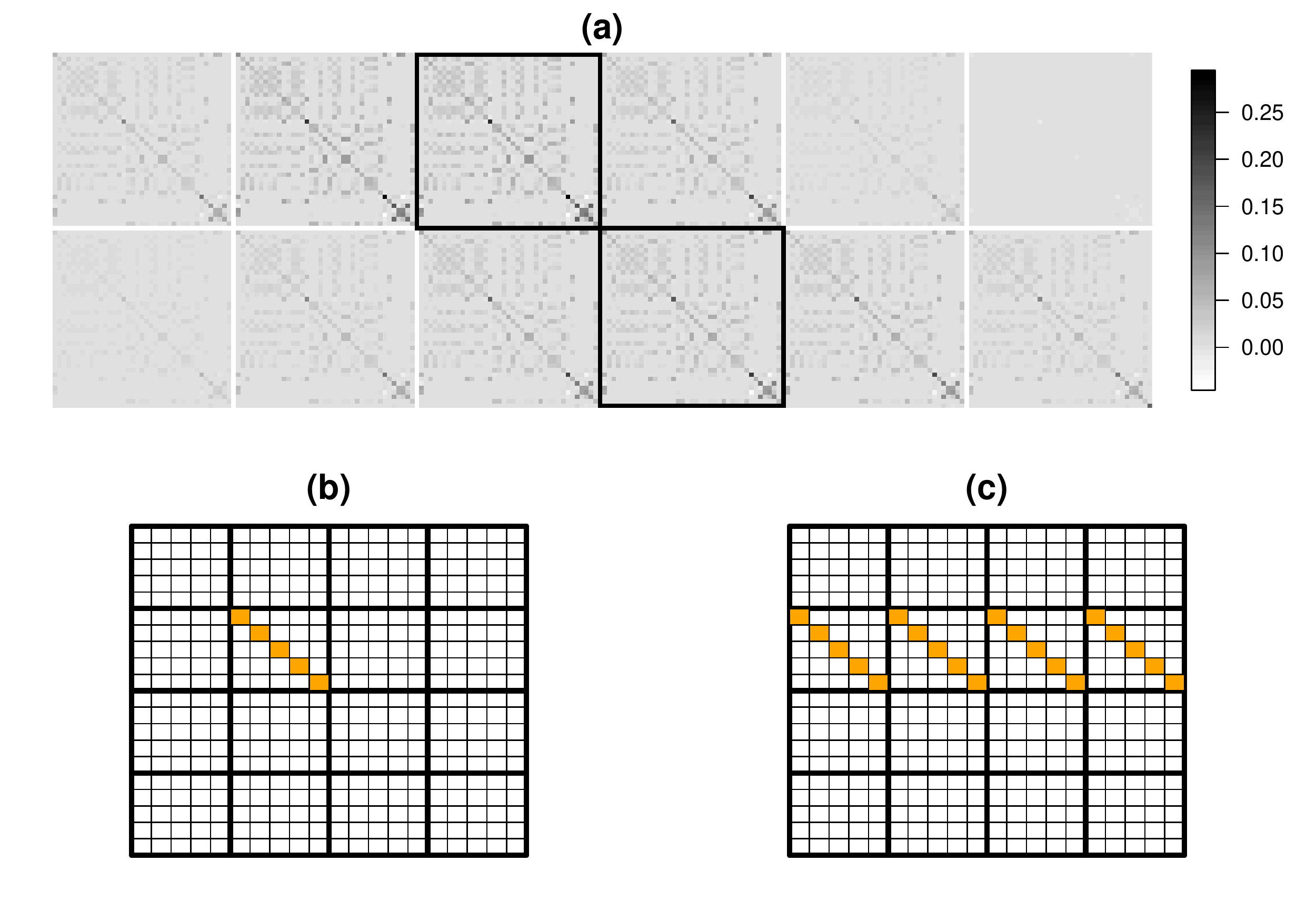} 
\caption{(a) Excerpt from the block hat matrix \eqref{blockhat} for a fit to the corpus callosum data, with diagonal blocks outlined in black. (b) Schematic illustration of the usual definition of df for the 2nd location, in a toy example with separate models at each of $L=4$ locations. (c) Proposed pointwise df for the 2nd location.}\label{hatfig}
\end{figure}

We can similarly define the \emph{pointwise leverage} of the $i$th observation at location $\ell$ as 
\[\sum_{\ell^*=1}^L\frac{\partial \hat{y}_{i\ell}}{\partial y_{i\ell^*}}=\sum_{\ell^*=1}^L h_{(\ell\ell^*)ii}.\]
This generalizes the ordinary leverage $h_{(\ell\ell)ii}$ for the $i$th observation in the $\ell$th-location model, given separate models for the $L$ locations. Pointwise leverage could be used to detect influential observations in functional-response regression, but we do not pursue this here. 

\subsection{Application to varying-coefficient models}\label{avc}
It must be acknowledged that our definition of pointwise df is not the only conceivable generalization of \eqref{odf} to account for sharing information across locations; and indeed it is not obvious how one might confirm that ours is the ``correct'' generalization. In this section, we provide a form of validation for Definition~\ref{pwdfdef}: namely, we show that it leads to the intuitively correct value in the case of varying-coefficient models.

Suppose we are given $n$ functional responses as in \eqref{thedata}, but the $i$th observation includes a predictor vector $\bx_i\in\mathbb{R}^p$ with $p<n$. Assume the functional responses arise from the varying-coefficient model
\begin{equation}\label{fosr}y_i(s)=\bx_i^T\bbeta(s)+\varepsilon_i(s)\end{equation}
\citep{ramsay2005}, also known as a ``function-on-scalar'' linear regression model \citep{reiss2010}. This setup is more general than that of the rest of the paper insofar as we are allowing multiple predictors; on the other hand, for $\bx_i=(1,t_i)^T$, our model is the restriction of \eqref{themod} to the case in which $f(t,s)$ is linear in $t$.
Let $\bX$ be the $n\times p$ design matrix with $i$th row $\bx_i^T$, and assume $\bbeta(s)=\bTheta\bb_s(s)$ where $\bb_s(s)$ is as in Section~\ref{setup} but now $\bTheta=(\theta_{ij})_{1\leq i\leq p, 1\leq j \leq K_s}$. Then \eqref{fosr} can be written in matrix form as $\bY=\bX\bTheta\bB_s^T+\bE$
[cf.\ \eqref{mateq}]. Letting $\btheta=\vecop(\bTheta)$ as before, we can estimate this coefficient vector by penalized GLS  as
\begin{eqnarray}\hat{\btheta}&=&\argmin_{\theta}\left[\{\by-(\bB_s\otimes\bX)\btheta\}^T(\hat{\bSigma}^{-1}\otimes \bI_n)\{\by-(\bB_s\otimes\bX)\btheta\}\right.\nonumber\\&&\qquad\qquad\qquad\qquad\qquad\left.+\btheta^T(\bP_s\otimes\bLambda)\btheta\right],\label{vccrit}\end{eqnarray}
\citep{ramsay2005,reiss2010}, where $\hat{\bSigma}^{-1}$ is a precision matrix estimate as in Section~\ref{polspgls} and $\bLambda=\mbox{Diag}(\lambda_1,\ldots,\lambda_p)$ [cf.\ \eqref{pglsvec}]. Perhaps more transparently, the penalty in \eqref{vccrit} can be written as $\sum_{k=1}^p\lambda_k\btheta_{k\cdot}^T\bP_s\btheta_{k\cdot}$ where $\btheta_{k\cdot}^T$ is the $k$th row of $\bTheta$, i.e., as the sum of separate penalties for the $p$ coefficient functions $\beta_k(s)=\btheta_{k\cdot}^T\bb_s(s)$, $k=1,\ldots,p$. The penalized OLS fit can be viewed as a special case of \eqref{vccrit} with $\hat{\bSigma}=\bI_L$.

The solution to \eqref{vccrit} yields, for given $s$, the mapping
\[\bx\mapsto E[y(s)|\bx]=\bx^T\hat{\bbeta}(s)=\bx^T\hat{\bTheta}\bb_s(s),\]
which is linear in $\bx$.
Intuitively, then, the pointwise df should equal the df of an ordinary linear regression with the same design matrix.   The following result shows that, under mild assumptions, Definition~\ref{pwdfdef} agrees with this expectation.
\begin{thm}\label{pwdfvc} Assume that $\bX$ is of rank $p$ and that \eqref{as1} and \eqref{as2} hold. Let $\mathbf{\cal H}$ be the hat matrix such that 
\begin{equation}\label{hvc}\hat{\by}=(\bB_s\otimes\bX)\hat{\btheta}=\mathbf{\cal H}\by,\end{equation}
where $\hat{\btheta}$ is given by \eqref{vccrit}. Then the pointwise df \eqref{dff} equals $d\el=p$ for $\ell=1,\ldots,L$.\end{thm}
The proof appears in Appendix~\ref{pwdfapp}. Having thus validated Definition~\ref{pwdfdef} for the case of varying-coefficient models, we turn next to evaluating the pointwise df of our estimators for varying-smoother models.

\subsection{Application to varying-smoother models}
Whereas the pointwise df is the same for $s_1,\ldots,s_L$ for varying-coefficient models, it varies with $s$ for varying-smoother models, and therefore can serve as a measure of the complexity of the fit $x\mapsto\hat{f}(\cdot,s)$ for different $s$. It is not self-evident from Definition~\ref{pwdfdef} that we can efficiently compute all $L$ pointwise df values, as opposed to computing $d_1,\ldots,d_L$ individually. But we now show that, for each of the methods of Sections~\ref{tpp} through \ref{2s}, there is indeed a readily computable expression for the entire pointwise df vector.
Recall that the pointwise df depends on the hat matrix $\mathbf{\cal H}$, which in turn is defined by $\hat{\by}=\vecop(\hat{\bY})=\mathbf{\cal H}\by$. Each of the following results provides the pointwise df vector $\bd$ corresponding to the given fitted values.

\begin{thm} \label{pp} \emph{(Pointwise df for tensor product penalty method)}
 Let $\hat{\by}=(\bB_s\otimes\bB_t)\hat{\btheta}$
where $\hat{\btheta}$ is the penalized OLS estimate \eqref{grr}, the penalized GLS estimate \eqref{pglsvec}, or the adaptively penalized OLS or GLS estimate in which \eqref{wct} is replaced by \eqref{ncpen}. Let
\[{\cal P}=\left\{\begin{array}{ll}\lambda_s(\bP_s\otimes\bQ_t)+\lambda_t(\bQ_s\otimes\bP_t), & \mbox{for penalty \eqref{wct}};\\\lambda_s(\bP_s\otimes\bQ_t)+\sum_{k=1}^{K_s^*}\lambda_{t,k}(\bQ^{b^*_k}_s\otimes\bP_t), & \mbox{for the adaptive penalty \eqref{ncpen}},\end{array}\right.\]
and let \[{\cal M}=\left\{\begin{array}{ll}{[(\bone_L^T \bB_s) \otimes \bB_t]} \left[(\bB_s^T\bB_s)\otimes (\bB_t^T \bB_t)+{\cal P}\right]^{-1}, & \mbox{for penalized OLS}; \\  {[(\bone_L^T\hat{\bSigma}^{-1} \bB_s) \otimes \bB_t]} \left[(\bB_s^T\hat{\bSigma}^{-1}\bB_s)\otimes (\bB_t^T \bB_t)+{\cal P}\right]^{-1}, & \mbox{for penalized GLS}.\end{array}\right.\]
Then the pointwise df vector is
\begin{equation}\label{tpd}\bd=(\bI_L\otimes\bone_n^T)[(\bB_s\otimes\bB_t)\odot(\bone_L\otimes{\cal M})]\bone_{K_s K_t}.\end{equation}  
\end{thm}

\begin{thm} \emph{(Pointwise df for the smoothed FPC score method)}\label{chiouthm}  
For the fitted values matrix $\hat{\bY}$ of Appendix~\ref{compchiou},  the pointwise df vector is
\begin{equation}\label{dfchiou}\bd=\bone_L+   \left[(\bB_s\bV_A)\odot(\bone_L\bone_{K_s}^T\bQ_s^T\bV_A)\right] \bM^{*T} (\bA^T_t\odot\bA^{cT}_t) \bone_n,\end{equation}
where $\bM^*=\left(\frac{1}{1+\lambda_a \tau_k}\right)_{1\leq k\leq K_t, 1\leq a\leq A}$, $\bA_t$ is given by \eqref{defa}, and $\bA^{c}_t=(\bI_n-\bJ_n)\bA_t$.
\end{thm}

\begin{thm}\label{thm2s}\emph{(Pointwise df for the two-step method)} 
Let $\hat{\bY}$ be the fitted values matrix for the two-step method, given by \eqref{fit2} for the penalized variant and in Appendix~\ref{fpcvs} for the FPC and penalized FPC variants. 
\begin{enumerate}[(a)]
\item\label{step1df} The vector of (ordinary) df for the initial fitted values matrix $\tilde{\bY}$ of \eqref{ytilde} is $\tilde{\bd}=\bM^T\bone_{K_t}$, where $\bM$ is given by \eqref{defm}.
\item \label{dfpen} The pointwise df resulting from the penalized variant of step~2 is
\begin{equation}\label{dsd}\bd=\bH_s \tilde{\bd}.\end{equation}
\item \label{dffpc} For the FPC and penalized FPC variants of step~2, 
 \begin{equation}\label{dfppp}\bd=\bone_L+\bB_s\bV_A\bN^{-1}\bV^T_A\bB^T_s(\tilde{\bd}-\bone_L),\end{equation}
where
\[\bN=\left\{\begin{array}{ll}\bV^T_A\bB_s^T\bB_s\bV_A, & \mbox{for the FPC variant};\\
\bV^T_A(\bB_s^T\bB_s+\lambda_s\bP_s)\bV_A, & \mbox{for the penalized FPC variant}.\end{array}\right.\]
\end{enumerate}
\end{thm}

Theorem~\ref{thm2s}(\ref{dfpen}) has the following interpretation: the effect on pointwise df of the penalized variant of step 2, in which we apply the function-direction smoother $\bH_s$ to the fitted value matrix, is simply to apply the same smoother to the pointwise df vector. Thus, if we think of $\tilde{\bd}$ as a vector of noisily measured complexity indices for $\{f(\cdot,s\el):\ell=1,\ldots,L\}$, then step~2 serves to denoise these measurements. The pointwise df vector \eqref{dfppp} for the other two variants of the two-step method is somewhat harder to interpret; but it can be shown that if $\bone_{L}$ belongs to the column space of $\bB_s\bV_A$, then   \eqref{dfppp} reduces to the form \eqref{dsd} with $\bH_s=\bB_s\bV_A\bN^{-1}\bV^T_A\bB^T_s$.

\section{Simulation study}\label{simsec}
\subsection{Simulation design}\label{simdesign}
For each of the simulation settings below we ran 100  replications, each with sample size $n=100$. The predictors $t_1,\ldots,t_{100}$ in each simulated data set belonged to the interval $[0,1]$; for computational reasons, 98 were sampled from the Uniform$[0,1]$ distribution, and the other two were taken to be 0 and 1.  For $i=1,\ldots,100$, the $i$th observed functional response was given by $y_i(s)=f(t_i,s)+\varepsilon_i(s)$ for $s\in S=\{0,\frac{1}{200},\ldots,1\}$, where $f$ was one of the two functions given below and
\begin{equation}\label{vee}\varepsilon_i(s)=\eta_i(s)+e_i(s),\end{equation}
where $\eta_i$ was drawn from a mean-zero Gaussian process with $\cov[\eta_i(s_1),\eta_i(s_2)]=\gamma\sigma^2(0.5^{200|s_1-s_2|})$, and  $e_i$ was drawn from a mean-zero Gaussian white noise process with variance $\sigma^2$; these two processes were mutually independent. One can think of $f(t_1,\cdot)+\eta_1(\cdot),\ldots,f(t_n,\cdot)+\eta_n(\cdot)$ as the true, imperfectly observed functional responses---i.e., $\eta_i(\cdot)$ is the deviation of the $i$th response function from the conditional mean function given $t_i$---whereas the $e_i(\cdot)$'s represent measurement error. \figref{funrespfig} displays an example set of functional responses.
\begin{figure}
\centering
\includegraphics[width=\textwidth]{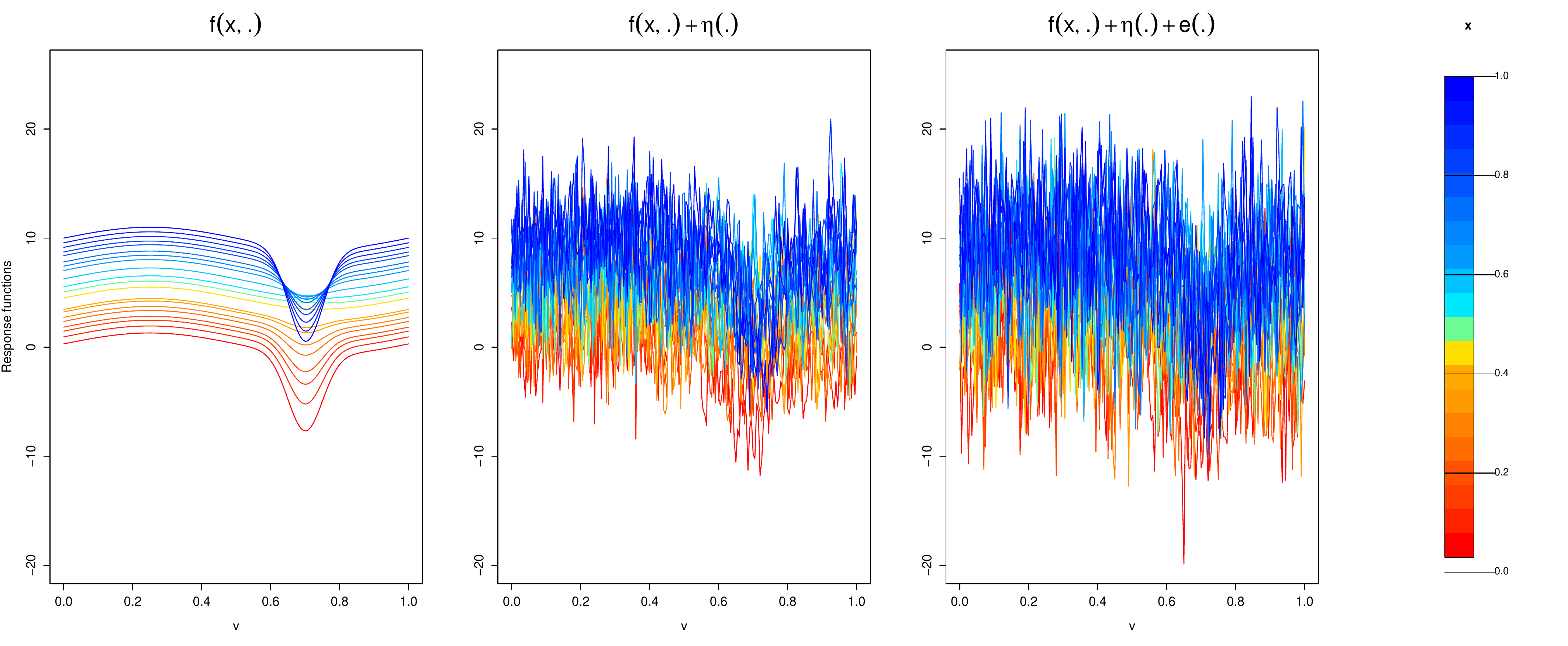}  
\caption{Simulated functional responses, color-coded by the value of the predictor $t$, with true mean function $f_2$, $R^2=0.3$, and $\gamma=1$. For different values of $\gamma$, the third subfigure would look similar, but the second would show greater variation among the functions for larger $\gamma$.}\label{funrespfig}
\end{figure}

There were eight simulation settings, based on two values for each of the following three factors.
\begin{enumerate}
\item The true mean function $f:[0,1]\times[0,1]\longrightarrow\mathbb{R}$ was either
\[
f_1(t,s) =8(s-.5)^2+\sin\left(\frac{\pi x[x(1-2p_s) + 2p_s^2 - 1]} {2p_s(p_s-1)}\right)
\]
where $p_s = \frac{\sin(2\pi s)+8}{16}$; or
\[f_2(t,s)=\sin(2\pi s) + 10x - \frac{\phi[20(s-0.7)]}{c_s^2}(t-c_s)^2\]  
where $c_s = 0.5-0.2(s-0.5)^2$ and $\phi$ is the standard normal density function. The function $f_1$ was chosen so that for each $s$, $f_1(\cdot,s)$ would have a single peak, at $t=p_s$; $f_2$ was designed to make $f_2(\cdot,s)$  quadratic in $t$, but approximately linear for $s$ far from 0.7---an example of the varying smoothness with respect to $t$ that our estimators aim to capture (see \figref{f1f2fig}).

\item The functional coefficient of determination \citep[cf.][]{muller2008} for the true mean function $f$ was set to 0.05 or 0.3. This parameter was defined as
 \begin{equation}\label{r2def}R^2=1-\frac{\sum_{i=1}^n \int_0^1 [y_i(s)-f(t_i,s)]^2 ds}{\sum_{i=1}^n\int_0^1[y_i(s)-\bar{y}(s)]^2 ds},\end{equation}
where $\bar{y}(s)=\sum_{i=1}^n y_i(s)/n$, and was evaluated approximately by replacing the integrals with a sum over the 201 points $s\in V$ (since these were equally spaced).   To attain the chosen value of $R^2$, the value of $\sigma^2$ was adjusted via the iterative procedure outlined in Appendix~\ref{r2app}.

\item The ratio $\gamma$ of the variance of the $\eta_i(\cdot)$ process to the variance of the $e_i(\cdot)$ process was set to 0.25 or 4.
\end{enumerate}

\begin{figure}
\centering
\includegraphics[width=.9\textwidth]{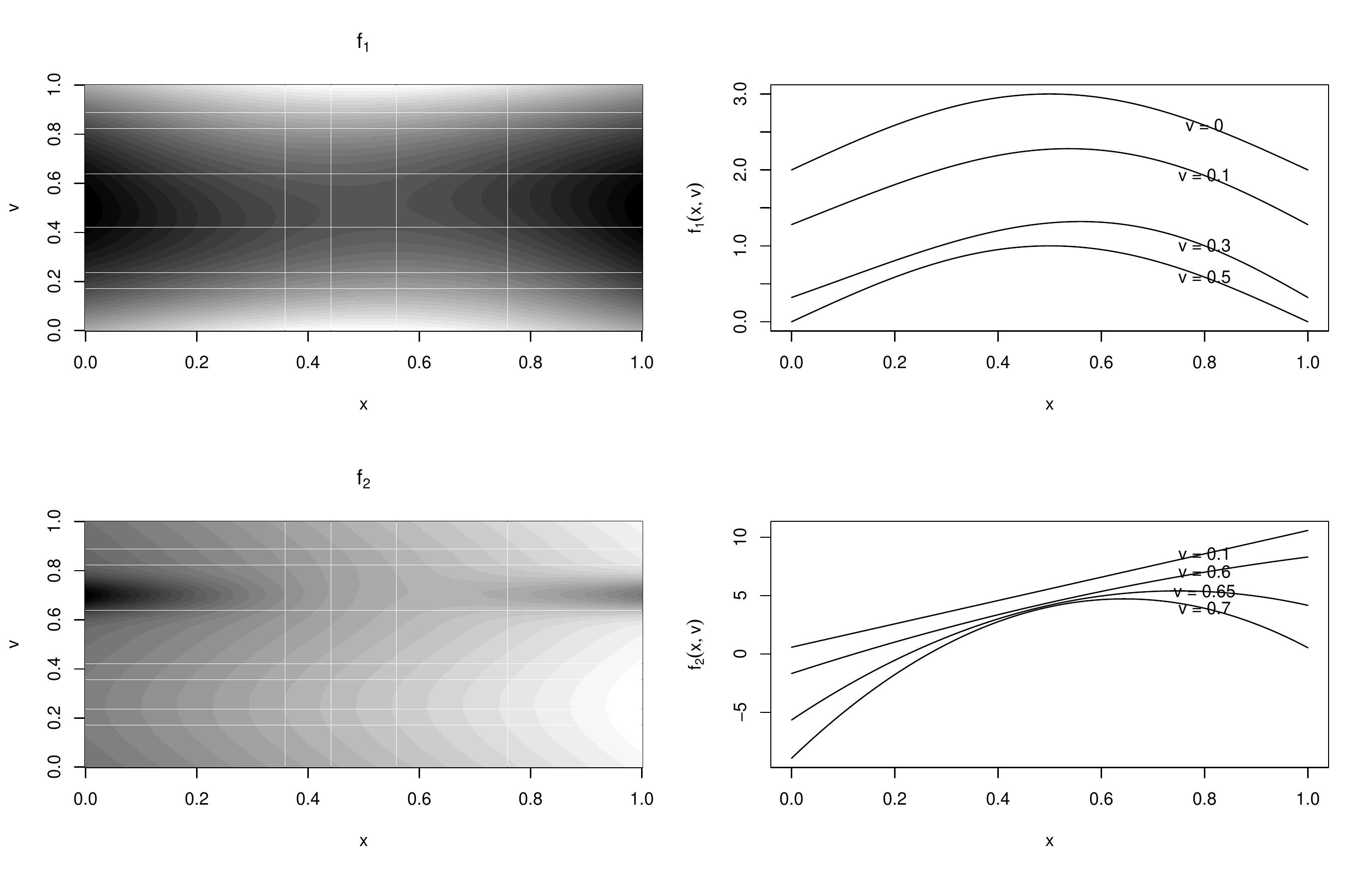} 
\caption{True mean functions used in the simulations, displayed both as grey-scale images and as cross-sections $f(\cdot,s)$ for selected $s$; $f_1$ is shown above, and $f_2$ below.}\label{f1f2fig}
\end{figure}

We compared the following methods:
\begin{enumerate}[(a)]
\item the OLS and GLS tensor product penalty methods of Section~\ref{tpp}, with either the non-adaptive penalty  \eqref{wct} (denoted by ``TP-OLS/GLS'') or the adaptive penalty \eqref{ncpen} (denoted by ``TP-OLS/GLS-adapt'');
\item the smoothed FPC score method of Section~\ref{fpc} (``FPC-scores''); and
\item the two-stage method of Section~\ref{2s}, with the three variants (penalized, FPC, and penalized FPC) of step~2 (``2s-pen''/``2s-FPC''/``2s-penFPC'').
\end{enumerate}
We used cubic $B$-spline bases with second-derivative penalties and dimensions $K_s=15$, $K_s=30$ and (for the adaptive penalty) $K_s^*=5$, with equally spaced knots and repeated knots at the endpoints as in \cite{ramsay2009}. For the computationally intensive penalized OLS and GLS methods, a slightly lower value of $K_s=25$ was used. The number of PCs was chosen by five-fold cross-validation for the FPC-scores and 2s-FPC methods, and fixed at 30 for the 2s-penFPC method. Five-fold cross-validation was also used to choose $\lambda_s$ for the 2s-pen and 2s-penFPC methods.

Performance was evaluated using both the integrated squared error (ISE) for estimates $\hat{f}$ of the mean function $f$,
\[\mbox{ISE}_f=\int_0^1\int_0^1[\hat{f}(t,s)-f(t,s)]^2\mbox{ }dt\mbox{ } ds,\]
and the ISE for estimating the derivative with respect to $t$,
\[\mbox{ISE}_{\partial f/\partial t}=\int_0^1\int_0^1\left[\frac{\partial\hat{f}}{\partial t}(t,s)-\frac{\partial f}{\partial t}(t,s)\right]^2\mbox{ }dt\mbox{ } ds.\]
The latter metric is particularly relevant for varying-smoother models, given our interest in estimating the shape of $f(\cdot,s)$ for each $s$.

\subsection{Results}
Figures~\ref{funcfig} and \ref{derivfig} present boxplots of relative ISE$_f$ and ISE$_{\partial f/\partial t}$ for all estimators, where relative ISE is defined as the ISE divided by the minimum ISE attained by any of the methods for a given replication.
Noteworthy results include the following:
\begin{enumerate}[(a)]
\item The tensor product penalty methods are generally best for $R^2=0.05$, especially with true function $f_1$; this is clearer for estimating $f$ than for estimating $\partial f/\partial t$.
 \item GLS generally outperforms OLS for $\gamma=4$, but the two are similar for $\gamma=0.25$. This is as expected since the higher $\gamma$ implies more strongly dependent residuals.
 \item The adaptive tensor product penalty seems more helpful for $f_2$ than for $f_1$, since the shape of $f_2(\cdot,s)$ depends more strongly on $s$.
 \item Of the three versions of step~2, the performance of the FPC variant is uniformly worst. The penalized FPC variant seems to have an advantage for estimating $\partial f_1/\partial t$ with $R^2=0.3$, but otherwise the penalized variant is the best of the three. Overall, the penalized variant of the two-step method has the best performance of all the methods for $R^2=0.3$.
 \item Separate smooths at each location (step~1 of the two-step approach) performed much less well than any of the methods shown. The relative ISE$_f$ was uniformly at least 6.8, while the relative ISE$_{\partial f/\partial t}$ was uniformly at least 3.6.
\end{enumerate}

\begin{figure}
\centering
\includegraphics[width=\textwidth]{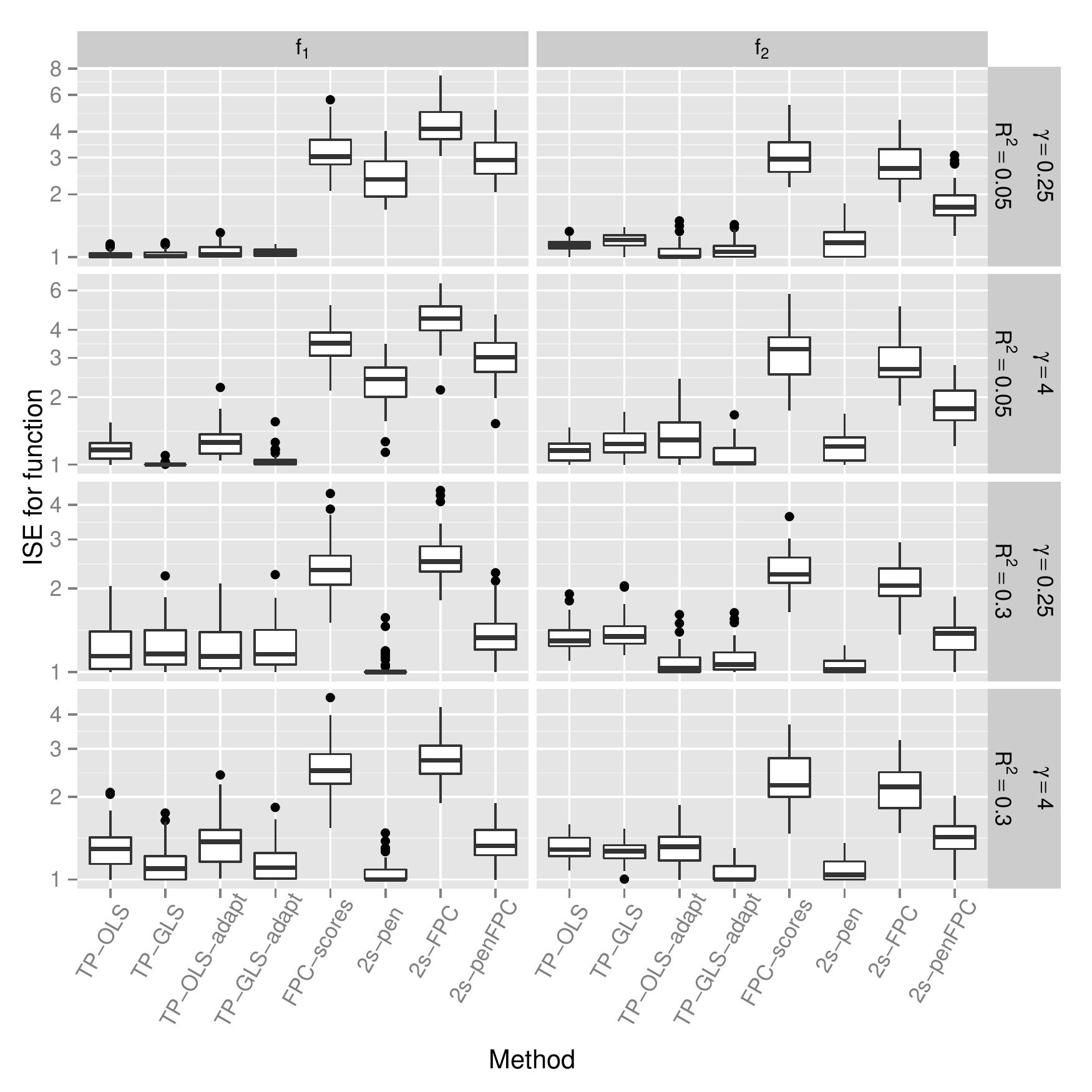}  
\caption{Relative ISE$_f$ for the different estimators (see the text for the abbreviations).}\label{funcfig}
\end{figure}

\begin{figure}
\centering
\includegraphics[width=\textwidth]{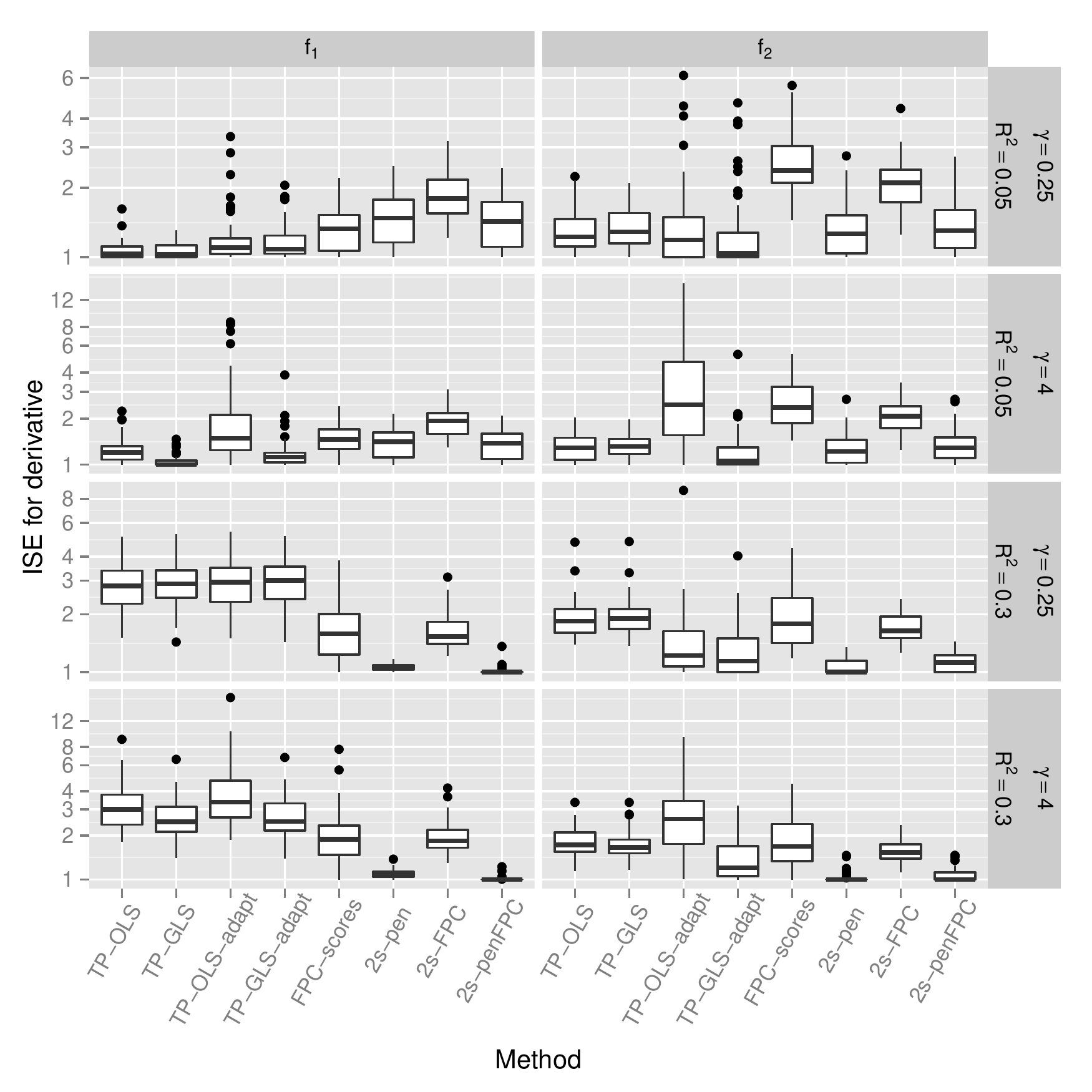}  
\caption{Relative ISE$_{\partial f/\partial t}$ for the different estimators.}\label{derivfig}
\end{figure}

The computation times for all methods are reported in Table~\ref{timetable}. The tensor product penalty methods are much slower than the other methods, due to the need to perform smoothing with $nL=20100$ observations. Although our implementation relies on the \texttt{bam} function in the \texttt{mgcv} package \citep{wood2006} for R \citep{R}, a function designed for very large data sets, choosing multiple smoothing parameters with this many observations remains computationally demanding.
\begin{table}
\centering
\begin{tabular}{|ll|ll|}
\hline
TP-OLS & 46.88 (18.91) & FPC-scores& 8.69 (3.20)\\TP-GLS & 92.09 (34.41) &  2s-pen&    4.11 (0.96)\\TP-OLS-adapt & 101.48 (30.90) & 2s-FPC& 5.00 (1.30)\\ TP-GLS-adapt  & 200.06 (54.03) & 2s-penFPC 
   & 5.76 (1.46) \\
\hline
\end{tabular}
\caption{Mean and standard deviation of the computation time, in seconds, for  each method in the simulation study.}\label{timetable}
\end{table}

\figref{simdffig} illustrates how the notion of pointwise df can help to explain the performance differences among methods. Recall that $f_2(\cdot,s)$ is quadratic, but approximately linear except near $v=0.7$. Accordingly, the separate pointwise smooths in the temporal direction, obtained in step~1 of the two-step method, have median df (shown as circles) near 3 for $v\approx 0.7$ and near 2 elsewhere, for the simulated data sets with $R^2=0.05$ and $\gamma=4$. As we would expect from Theorem~\ref{thm2s}(b), the penalized variant of step~2 yields pointwise df values near these intuitively ``correct'' values. Similarly, the pointwise df for GLS with adaptive tensor product penalty rises to a peak for roughly the same range of $s$ values, whereas the pointwise df vector for the non-adaptive penalty is quite flat. The pointwise df for the smoothed FPC scores estimator fluctuates with $s$ in a manner seemingly unrelated to the true shape of $f_2(\cdot,s)$. These observations provide some insight into the comparatively poor performance of the non-adaptive GLS and smoothed FPC scores estimators.
\begin{figure}
\centering
\includegraphics[width=.8\textwidth]{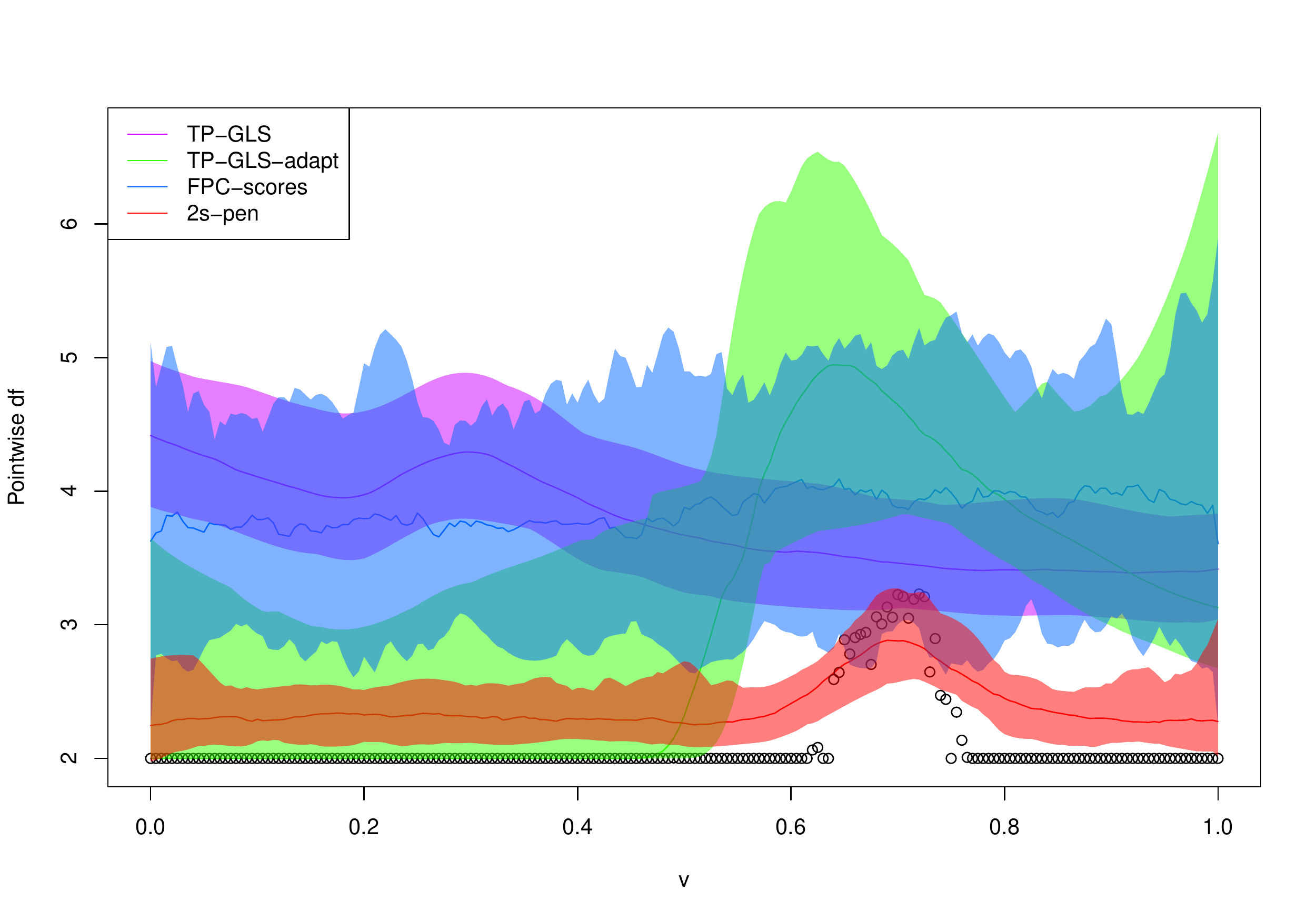}
\caption{Median and 5th and 95th percentiles of the pointwise df attained by the indicated methods, for the simulated data with true mean function $f_2$, $R^2=0.05$, and $\gamma=4$. Circles indicate medians for separate pointwise smooths (initial estimates for the two-step method).}\label{simdffig}
\end{figure}


\section{Application: Development of corpus callosum microstructure}\label{realsec}
We now return to the corpus callosum fractional anisotropy data described in the introduction.  A full description of the image data processing, as well as an illuminating set of analyses, can be found in \cite{imperati2011}. Here we aim to estimate the mean FA $f(t,s)$ where $t$ denotes age and $s$ denotes location, expressed as arc length along the path depicted in \figref{stan}. Before presenting our modeling results let us consider some evidence for nonlinear change in mean FA with respect to age. For $\ell=1,\ldots,107$ we performed a restricted likelihood ratio test \citep[RLRT;][]{crainiceanu2004} to test the null hypothesis that $f(\cdot,s\el)$ is linear (mean FA for the $\ell$th voxel changes linearly with age) against the alternative of nonlinear change. \figref{nonlinfig} shows the resulting $p$-values for each voxel. (These $p$-values are not adjusted for multiple tests, since our aim here is descriptive rather than to test the global null hypothesis that $f(\cdot,s\el)$ is linear for each $\ell$.) Also shown are separate penalized spline smooths for FA at three voxels, with the smoothing parameter chosen by REML. The first of these voxels is located in the sensorimotor portion of the brain; the second, in the posterior portion of the prefrontal lobes, with projections into the inferior and middle frontal gyri; and the third, in the anterior portion of the prefrontal lobes. In the first and third voxels, it appears that mean FA attains a peak in young adulthood, and then declines. These nonlinear trajectories are consistent with the RLRT $p$-values of .019 and .001, respectively, for the two voxels. But for the second voxel, there is no strong evidence for nonlinear change in mean FA. Features of these curves, such as peaks, are of biological interest, as they may provide insight into characteristic patterns of development for different cognitive abilities. However, the separate smooths at each voxel are quite noisy. Our hope is that varying-smoother models can improve the curve estimates by sharing information across neighboring voxels.
\begin{figure}
\centering
\includegraphics[width=.9\textwidth]{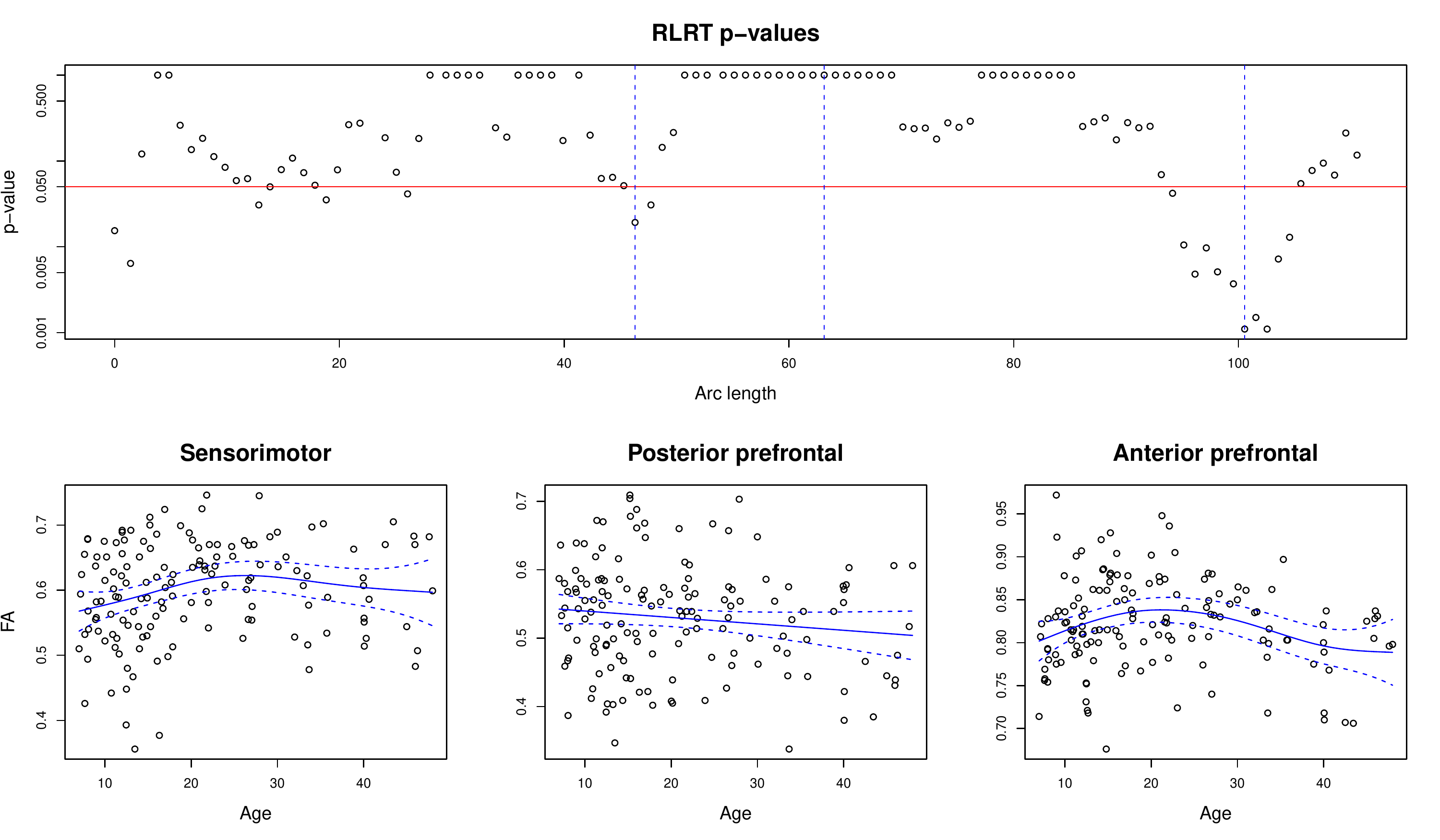} 
\caption{Above, $p$-values from restricted likelihood ratio tests of the null hypothesis that mean FA changes linearly with age. Below, curve estimates, $\pm 2$ approximate standard errors, for the three voxels indicated above by blue dashed vertical lines; headings refer to the brain regions where these voxels are located.}\label{nonlinfig}
\end{figure}

We estimated the mean FA $f(t,s)$ by the eight methods that were included in the simulation study. For these data, neither adaptive penalization (for the tensor product penalty approach) nor the three variants of step~2 (for the two-step approach) produced noteworthy differences; so we focus mainly on non-adaptive OLS and GLS, and on the penalized FPC variant of step~2. \figref{imageplotfig} presents the four estimates of $f(t,s)$. In general, FA varies much more with respect to $s$ (between brain regions) than with respect to $t$ (at different age levels for a given region). The OLS estimate appears undersmoothed, so that $\hat{f}(\cdot,s)$ tends to include spurious bumps with respect to age. But for the other three methods, variation with age is so thoroughly drowned out as to be nearly imperceptible.
\begin{figure}
\centering
\includegraphics[width=\textwidth]{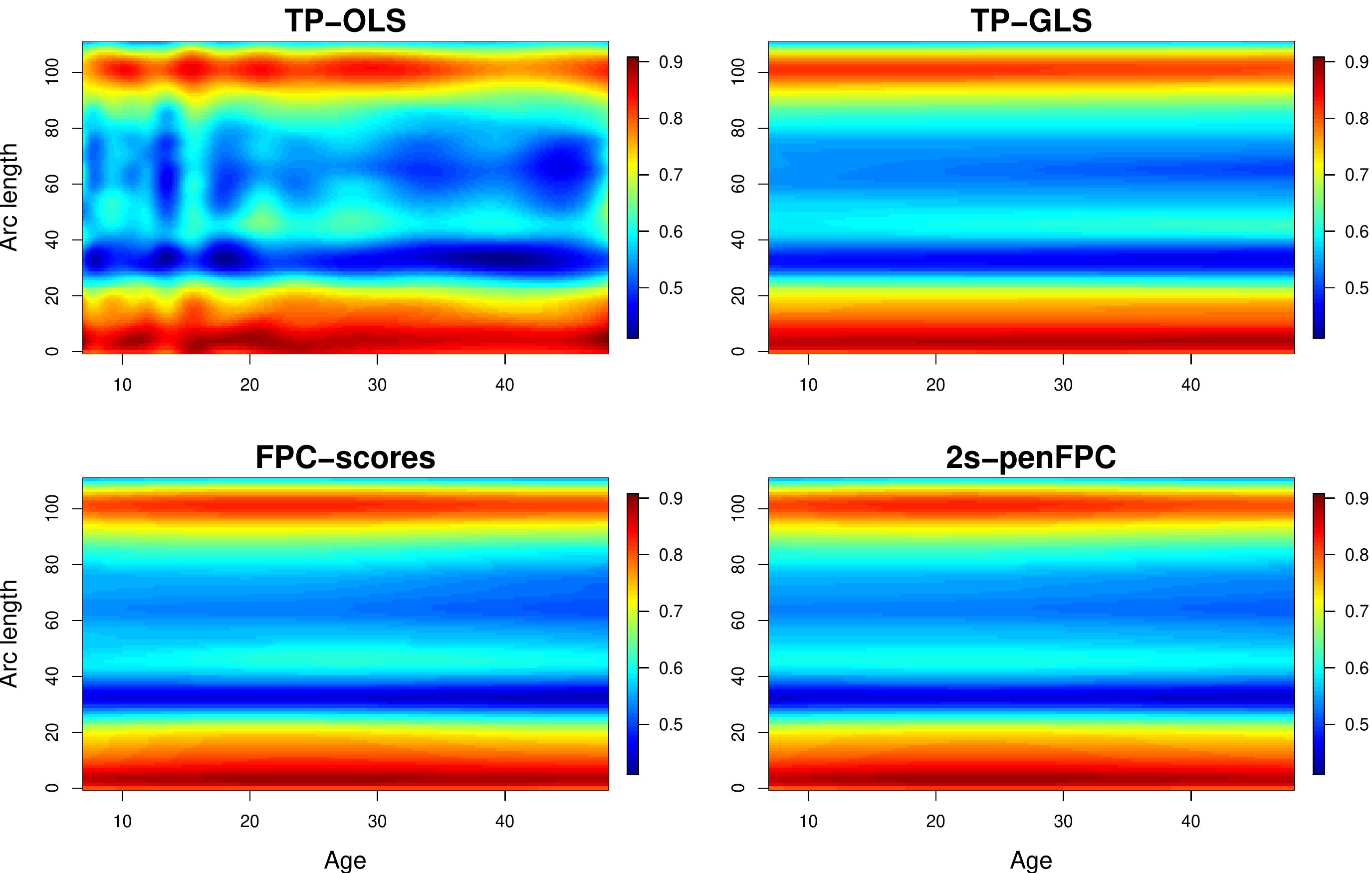} 
\caption{Estimates of the mean function $f(t,s)$, where $t$ denotes age and $s$ denotes arc length.}\label{imageplotfig}
\end{figure}

The rainbow plots \citep{hyndman2010} in \figref{rainbowfig} make it easier to examine the shape of the estimates $\hat{f}(\cdot,s)$ for different methods and different locations $s$. Vertical lines are drawn at the same three voxels as in \figref{nonlinfig}, and vertical progression from blue to red at a given arc length indicates that the estimate $\hat{f}(\cdot,s)$ is monotonic for that $s$. Recall that separate nonparametric regressions suggested that FA attained a peak for the first and third, but not the second. Although varying-smoother models aim to improve on separate nonparametric regressions at each voxel, it seems reasonable to expect that the true shape of $f(\cdot,s)$ is broadly consistent with the results of the separate models. 
\begin{figure}
\centering
\includegraphics[width=\textwidth]{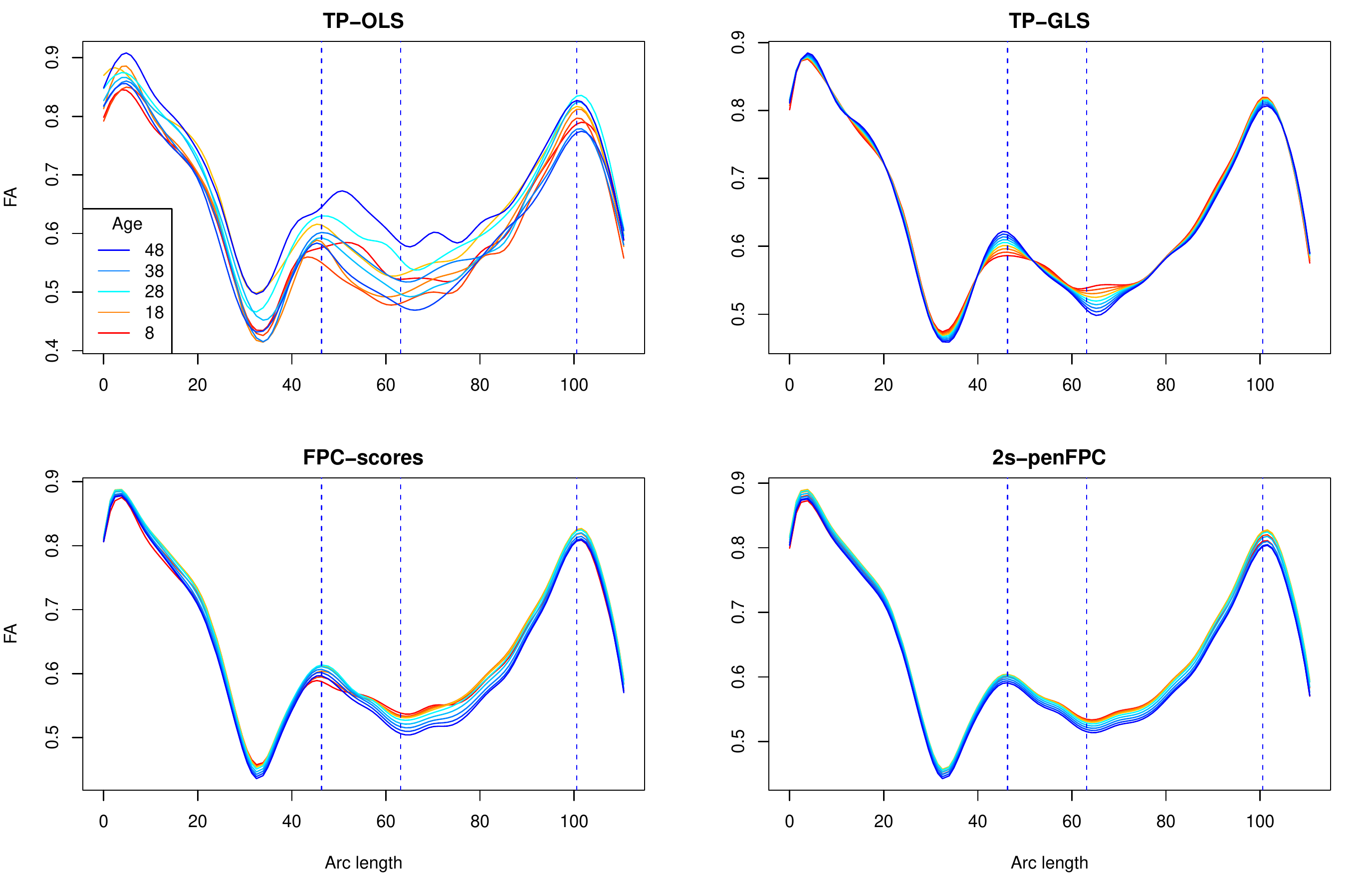} 
\caption{Rainbow plots for fitted value functions $\hat{f}(t,\cdot)$ with $t=8,13,18,\ldots,48$.}\label{rainbowfig}
\end{figure}

As  in \figref{imageplotfig}, the penalized OLS estimates $\hat{f}(\cdot,s)$ are very erratic, with apparently spurious fluctuations with respect to age. The smoothed-FPC-scores and two-stage estimates appear consistent with the scatterplots in \figref{nonlinfig}: they indicate that mean FA peaks in young adulthood in the sensorimotor and anterior prefrontal regions displayed there, but decreases linearly with age in the posterior prefrontal region. The penalized GLS estimates $\hat{f}(\cdot,s)$, on the other hand, indicate that mean FA changes monotonically in all three regions---suggesting that penalized GLS is oversmoothing with respect to age.

These impressions are borne out by \figref{dfr}(a), in which the pointwise df for penalized OLS is seen to be uniformly low while that for penalized GLS is uniformly high. In line with Theorem~\ref{thm2s}(\ref{dfpen}), the pointwise df for the penalized variant of the two-step method is a smoothed version of the df for the separate smooths. But all three variants of the two-step approach, as well as the FPC scores method, have pointwise df values that vary, reassuringly, within the same range as the ordinary df for 107 separate smooths. 
\begin{figure}
\centering
\includegraphics[width=1.05\textwidth]{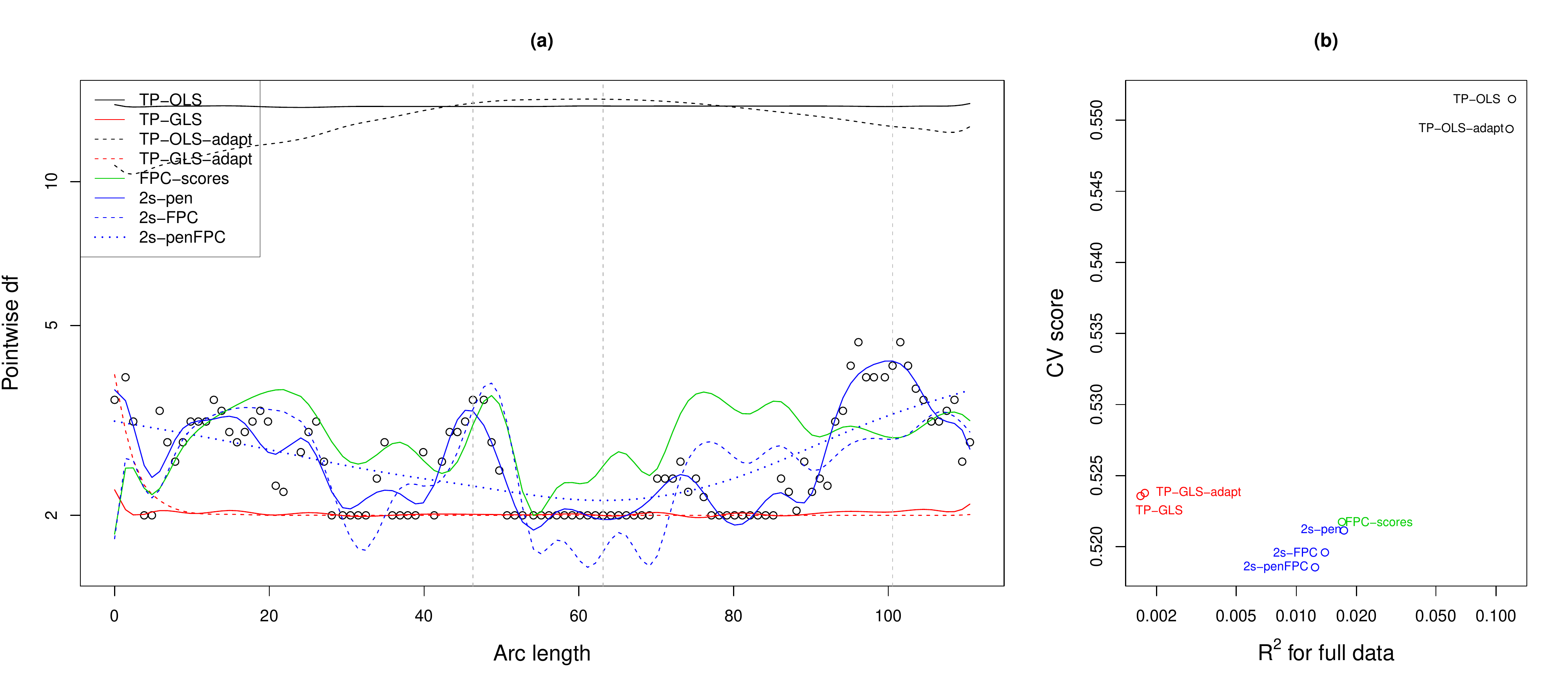} 
\caption{(a) Pointwise df for the methods applied to the corpus callosum FA data. Circles denote df for separate smooths at each of the 107 voxels. (b) $R^2$, given by a trapezoidal approximation to \eqref{r2def}, plotted against prediction error estimate from repeated five-fold cross-validation.}\label{dfr}
\end{figure}

\figref{dfr}(b) plots functional $R^2$ values  for the eight methods against prediction error estimates, based on repeated five-fold cross-validation \citep{burman1989} with 10 different splits of the observations into five validation sets. This figure provides further evidence that the penalized OLS and GLS approaches overfit and underfit the data, respectively. The $R^2$ values for the two versions of penalized OLS are approximately 0.12, while those for penalized GLS are about 0.002. The FPC-scores and two-step methods have $R^2$ values between these extremes (in the 0.013-0.017 range), and attain better predictive performance.

\figref{pref} presents the estimates $\hat{f}(\cdot,s\el)$, for $\ell=91,\ldots,98$, of the mean FA as a smooth function of age. These eight voxels correspond to arc lengths 94.1--101.5, a range whose right terminus corresponds to the rightmost peak observed in \figref{rainbowfig}, and to the most prominent trough in the $p$-value plot of \figref{nonlinfig}. The upper left subfigure displays separate curve estimates of FA for the eight voxels. In line with the results shown in Figures~\ref{imageplotfig} through \ref{dfr}, the penalized OLS and GLS curves clearly overfit and underfit, respectively, whereas the FPC-scores and two-step estimates appear to do a reasonable job of ``denoising'' the separate curve estimates.
\begin{figure}
\centering
\includegraphics[width=.82\textwidth]{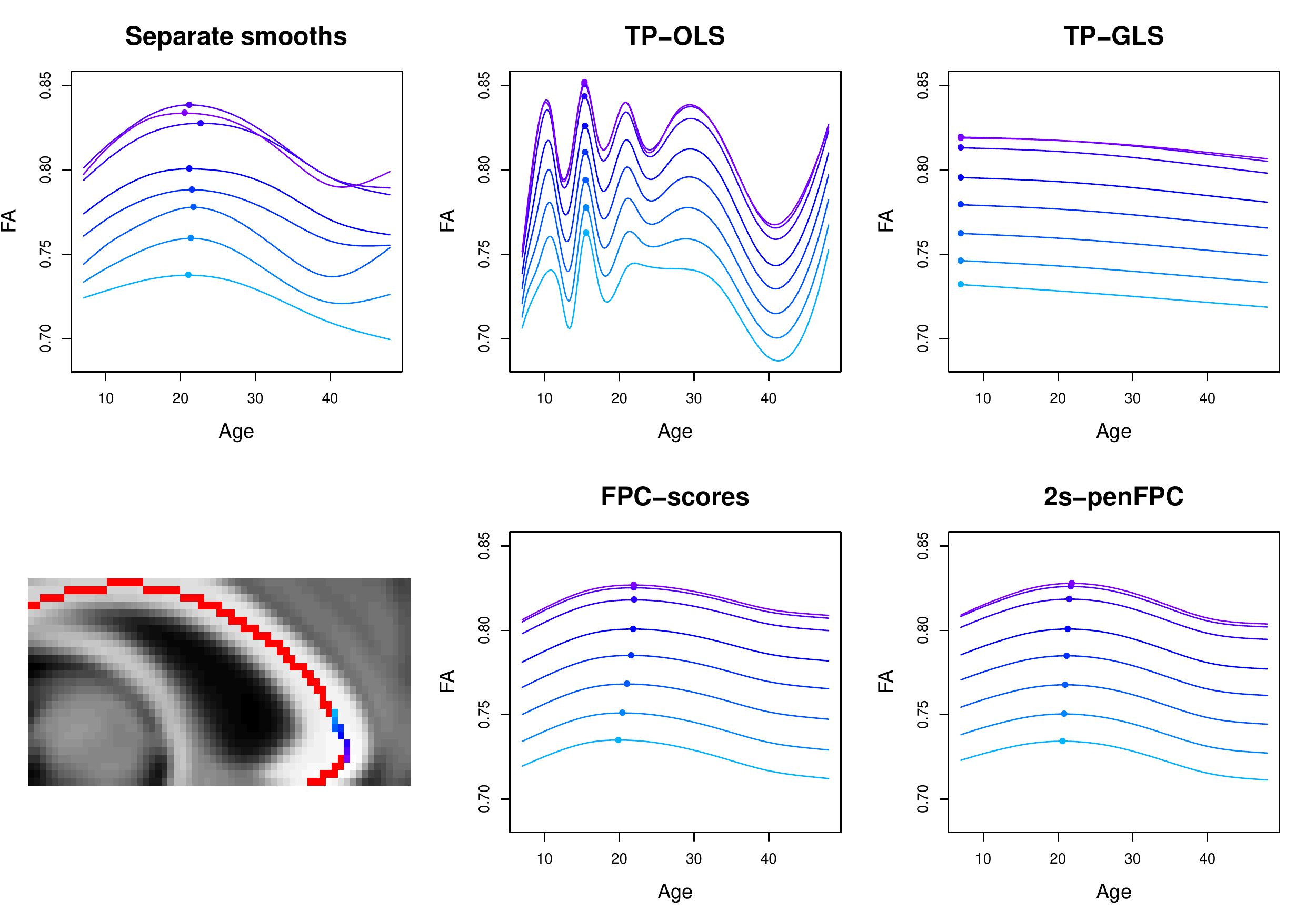} 
\caption{FA as a function of age for eight voxels in the prefrontal cortex (shown at lower left using the same blue-to-violet color scheme as the curves), as estimated by the indicated methods. Dots indicate estimated age of peak FA.}\label{pref}
\end{figure}

\section{Discussion}\label{discsec}
In this paper we have developed three approaches to fitting varying-smoother models with functional responses, and have introduced pointwise degrees of freedom, a tool for characterizing different model fits in this setting. We have focused on the case in which the function domain $\cal S$ is a finite interval on the real line. Future work will consider varying-smoother model methodology for more general $\cal S$, in particular ${\cal S}\subset \mathbb{R}^3$, as often occurs in neuroimaging applications. Linear models with spatially varying coefficients have been considered by a number of authors for this domain of application \citep[e.g.,][]{tabelow2006,smith2007,brezger2007,heim2007,li2011}, but little work, if any, has focused on  general smooth pointwise effects of predictors at different brain locations in multidimensional space. Varying-smoother modeling is by no means restricted to brain imaging applications, and we anticipate that it may be usefully applied to functional responses arising in many fields.

We have focused here on the case of a single scalar predictor. Further work is needed to extend both the modeling methodology and the definition of pointwise degrees of freedom to  multiple predictors and predictors that vary with $s$. We also would like to derive small- and large-sample error rates for the proposed estimators, which could provide insight into the disparate patterns of relative performance that we have observed under different scenarios.

We have restricted consideration here to point estimation of the mean function $f$. Interval estimation requires care even for simple nonparametric regression \citep{wood2006CI}; even more so for the more complex bivariate smoothers we have presented. Appendix~\ref{intapp} outlines approximate confidence interval methodology for the penalized two-step method of Section~\ref{penvt}, but much more research is needed on interval estimates for varying-smoother models.

Code for the methods discussed above is available for the authors, and we plan to disseminate some of the functions via the R package \texttt{refund} \citep{refund}, which is available on the CRAN repository (http://cran.r-project.org/web/\-packages/\-refund).

\subsubsection*{Acknowledgments}
The authors thank the reviewers in advance for their efforts, and Yin-Hsiu Chen, Ciprian Crainiceanu, Jeff Goldsmith, Lan Huo, Mike Milham, Todd Ogden, Eva Petkova, David Ruppert, Fabian Scheipl, and Simon Wood for very helpful advice and feedback. The first author's work was supported by National Science Foundation grant DMS-0907017, and the work of the first three authors was supported by National  Institutes of Health grant 1R01MH095836-01A1.

\appendix
\setcounter{equation}{0}
\renewcommand{\theequation}{A.\arabic{equation}}

\section{Banded inverse covariance estimate}\label{cbp}
For the estimate $\hat{\bSigma}^{-1}$ in the penalized (feasible) GLS criterion \eqref{pgls}, we use the banded variant \citep{bickel2008} of the modified Cholesky decomposition of the precision matrix \citep{pourahmadi1999}. This precision matrix estimate has the form $\hat{\bSigma}^{-1}=\bT^T\bD\bT$ where $\bD$ is a diagonal matrix with nonnegative diagonal entries and $\bT$ is a $k$-banded lower triangular matrix, ensuring that $\hat{\bSigma}^{-1}$ is $k$-banded and positive semidefinite. A banded precision matrix implies, under normality, that values at two distant locations along the function are conditionally independent, given the values at all other locations---a reasonable assumption for many, albeit not all, functional data sets. 

For choosing the number of bands $k$, \cite{bickel2008} propose a resampling procedure.  Here we choose $k$ by a new procedure that obviates the need for resampling.
Our method  is based on the work of \cite{ledoit2002} on high-dimensional sphericity tests. Let $\bS$ be the $p\times p$ sample covariance matrix for an $n\times p$ data matrix $\bX$. By Proposition~3 of \cite{ledoit2002}, 
\begin{equation}\label{lwstat}\frac{1}{2}[n p \mbox{ }\tr(\bS^2)/(\tr\bS)^2 - n - p - 1]\end{equation} is approximately standard normal for large $n,p$.
Our idea is to compute \eqref{lwstat} with $\bX$ taken to be the whitened residual matrix $(\bY-\bB_t\hat{\bTheta}\bB_s^T)\hat{\bSigma}^{-1/2}$, where $\hat{\bTheta}$ is the penalized OLS estimate \eqref{pols}, and $\hat{\bSigma}^{-1}$ is the $k$-banded estimate of \cite{bickel2008} for each of a range of values of $k$. Large positive values of \eqref{lwstat} indicate that multiplication by a $k$-banded square-root precision matrix is inadequate to remove the residual dependence, whereas large negative values signal ``overwhitening,'' i.e., the residual vectors exhibit smaller sample covariances than would typically arise by chance. We choose the value of $k$ for which the magnitude of \eqref{lwstat} is smallest, which generally seems to be a good compromise between these extremes. We have not studied how this criterion for choosing $k$ compares with the resampling method of \cite{bickel2008} for estimation of $\bSigma^{-1}$, but that is not the goal here. Rather, we need to transform correlated residuals to approximately whitened residuals for penalized GLS, and our proposal offers a means to that end that avoids computationally intensive tuning parameter selection.

\section{Smoothing parameter selection}\label{spsapp}
The penalized OLS minimization \eqref{grr} is a generalized ridge regression problem, for which automatic criteria  for choosing the tuning parameters $\lambda_s,\lambda_t$ \citep{reiss2009} can be readily optimized with the \texttt{mgcv} package  \citep{wood2006,wood2011} for R \citep{R}. However, application of criteria such as REML and generalized cross-validation \citep{craven1979}  to \eqref{grr} presupposes that the components of $\by$ are conditionally independent given $\btheta$---an untenable assumption here, as noted in \eqref{wfd}. To take within-function dependence into account when fitting  varying-\emph{coefficient} models with functional responses (see Section~\ref{avc}), \cite{ramsay2005} recommend choosing the smoothing parameters by leave-one-function-out cross-validation \citep{rice1991}. 

On the other hand,
\cite{krivobokova2007} have shown that REML-based smoothness selection is quite robust to correlated errors. Consistent with this, some authors \citep[e.g.,][]{crainiceanu2012} have reported good performance of REML-based smoothing in functional-response analyses that treat all the residuals as independent. Moreover, unlike REML, cross-validation remains difficult  with multiple smoothing parameters. Hence, in Sections~\ref{simsec} and \ref{realsec}, we examine the performance of penalized OLS with REML-based smoothness selection, ignoring the within-function dependence.

In penalized GLS, we attempt to remove within-function dependence by prewhitening. The minimization problem \eqref{pglsvec} 
is tantamount to penalized OLS for response vectors $\hat{\bSigma}^{-1/2}\by_{1\cdot},\ldots,\hat{\bSigma}^{-1/2}\by_{n\cdot}$, for which the within-function covariance, conditional on $x_i$, is approximately $\bI_L$. Thus the $nL$ residuals, in the normal mixed model representation underlying REML selection of  $\lambda_s,\lambda_t$, can reasonably be viewed as independent and identically distributed \citep[cf.][]{reiss2010}.

\section{Tensor product penalty derivations: Proofs of Theorems~\ref{tpen} and \ref{vdtpen}}\label{pt1}
\subsection{Proof of Theorem~\ref{tpen}} By \eqref{tpdef}, for each $v$, $f(\cdot,s)$ is of form (\ref{betabx}) with $\bgamma=\bTheta\bb_s(s)$.  Thus $r_t[f(\cdot,s)]=\bb_s(s)^T\bTheta^T\bP_t\bTheta\bb_s(s)$ and
\begin{eqnarray*}\int_{\cal S} r_t[f(\cdot,s)]dv & = & \int_{\cal S} \tr\left[\bb_s(s)\bb_s(s)^T\bTheta^T\bP_t\bTheta\right]dv \\
& = &\tr\left[\int_{\cal S} \bb_s(s)\bb_s(s)^Tds\bTheta^T\bP_t\bTheta\right] \\
& = & \tr(\bQ_s\bTheta^T\bP_t\bTheta).\end{eqnarray*}
Using the identity $\tr(\bK^T\bL)=(\vecop\bK)^T(\vecop\bL)$ and standard results for Kronecker products, we obtain
\begin{eqnarray*}\int_{\cal S} r_t[f(\cdot,s)]ds & = & \vecop(\bTheta\bQ_s^T)^T\vecop(\bP_t\bTheta) \\ 
& = & [(\bQ_s\otimes\bI_{K_t})\btheta]^T[(\bI_{K_s}\otimes\bP_t)\btheta] \\
& = & \btheta^T(\bQ_s\otimes\bP_t)\btheta.\end{eqnarray*}
Arguing similarly for $\int_{\cal T} r_s[f(t,\cdot)]dt$ and substituting into (\ref{penf}) yields \eqref{wct}.

\subsection{Proof of Theorem~\ref{vdtpen}} Generalizing the proof of Theorem~\ref{tpen}, we have
\begin{eqnarray*}\int_{\cal S} b^*_k(s)r_t[f(\cdot,s)]ds & = & \int_{\cal S} \tr\left[b^*_k(s)\bb_s(s)\bb_s(s)^T\bTheta^T\bP_t\bTheta\right]ds \\
& = &\tr\left[\int_{\cal S} b^*_k(s)\bb_s(s)\bb_s(s)^Tds\bTheta^T\bP_t\bTheta\right] \\
& = & \tr(\bQ^{b^*_k}_s\bTheta^T\bP_t\bTheta)\\
& = & \btheta^T(\bQ^{b^*_k}_s\otimes\bP_t)\btheta.\end{eqnarray*}
Thus, replacing $\lambda_t$ by \eqref{smallbasis} converts penalty \eqref{wct} to \eqref{ncpen}.

\section{Computational details for the FPC-based methods}\label{compfpc}
\subsection{Smoothed FPC scores method}\label{compchiou}
An approximate matrix equation for the Karhunen-Lo\`{e}ve expansion \eqref{kle} is given by
\begin{equation}\label{ame}\bY\approx\bone_n\hat{\bmu}^T + \hat{\bC}\bV_A^T\bB_s^T,\end{equation}
where  $\hat{\bmu}\in\mathbb{R}^L$ is an estimate of the discretized mean function, and $\bV_A=(\bv_1\ldots\bv_A)$ with $\bv_1,\ldots,\bv_A$ chosen so that $\bv_a^T\bb_s(\cdot)$ is an estimate of $\phi_a(\cdot)$.
 To obtain the required estimates in \eqref{ame}, and thereby estimate model \eqref{chioumod}, we proceed as follows:
\begin{enumerate}[(i)]
\item As a standard presmoothing step for functional data \citep{ramsay2005, ramsay2009}, project the rows of the raw response matrix $\bY$ onto the span of the $v$-direction basis to obtain $\bY\bB_s(\bB_s^T\bB_s)^{-1}\bB_s^T$. Then take the simple mean function estimate 
\begin{equation}\label{mfe}\hat{\bmu}^T=\frac{1}{n}\bone_n^T\bY\bB_s(\bB_s^T\bB_s)^{-1}\bB^T_s.\end{equation} 
\item By the argument of \cite{ramsay2005} adapted to our notation, if $\bu_a$ is the $a$th leading eigenvector of 
\[n^{-1}\bQ^{1/2}_s(\bB_s^T\bB_s)^{-1}\bB_s^T\bY^T(\bI_n-\bJ_n)\bY\bB_s(\bB_s^T\bB_s)^{-1}\bQ^{1/2}_s,\]
 then the $a$th estimated PC function is given by $\hat{\phi}_a(\cdot)=\bv_a^T\bb_s(\cdot)$ where $\bv_a=\bQ_s^{-1/2}\bu_a$.\
 The estimated PC scores are given by 
 \begin{equation}\label{epcs}\hat{\bC}=(\hat{\bc}_1\ldots\hat{\bc}_A)=(\bI_n-\bJ_n)\bY\bB_s(\bB_s^T\bB_s)^{-1}\bQ_s\bV_A.\end{equation} 
 \item For $a=1,\ldots,A$, we view the $a$th PC score as a smooth function $g_a(t)$ of $t$, and fit the smooth $\hat{\bg}_a=[\hat{g}_a(t_1),\ldots,\hat{g}_a(t_n)]^T=\bB_t(\bB_t^T\bB_t+\lambda_a\bP_t)^{-1}\bB_t^T\hat{\bc}_a$, where $\lambda_a$ is chosen to optimize the REML criterion.   Let $\hat{\bG}=(\hat{\bg}_1\ldots\hat{\bg}_A)$. 
\item The fitted values are \begin{eqnarray}\nonumber\hat{\bY}&=&\bone^T_n\hat{\bmu}^T + \hat{\bG}\bV_A^T\bB_s^T\\&=&\bJ_n\bY\bB_s(\bB_s^T\bB_s)^{-1}\bB_s^T + \hat{\bG}\bV_A^T\bB_s^T.\label{yhatchiou}\end{eqnarray}
\end{enumerate}
Steps~(i) and (ii) can be implemented using the functions \texttt{Data2fd} and \texttt{pca.fd} of the R package \texttt{fda} \citep{ramsay2009}. (An alternative implementation of functional PCA is available in the \texttt{refund} package \citep{refund}.) We do not impose roughness penalties here, as these  would require tuning by cross-validation, which would likely offer minimal benefit for these intermediate steps. We do, however, use cross-validation to choose $A$, the number of FPCs.

\subsection{FPC variants of the two-step method}\label{fpcvs}
The (unpenalized) FPC variant of step~2 (Section~\ref{fpcvt}) proceeds as follows:
\begin{enumerate}[(i)]
\item Similar to step (i) in Section~\ref{compchiou}, we begin with a light presmoothing step of projecting each row of $\tilde{\bY}$ onto the span of the $v$-direction basis, resulting in the $n\times L$ matrix $\tilde{\bY}\bB_s(\bB_s^T\bB_s)^{-1}\bB^T_s$.
\item Each row of that matrix is decomposed into the discretized estimated mean function \eqref{mfe}    and a deviation from the mean function, to obtain
\begin{equation}\label{dcomp}\bJ_n\bY\bB_s(\bB_s^T\bB_s)^{-1}\bB^T_s+(\tilde{\bY}-\bJ_n\bY)\bB_s(\bB_s^T\bB_s)^{-1}\bB^T_s.\end{equation}
\item The rows of the second matrix in \eqref{dcomp} are projected onto the span of the leading estimated PC functions $\hat{\phi}_1,\ldots,\hat{\phi}_A$ of the raw response data. In the notation of Section~\ref{compchiou}, this projection is perfomed by postmultiplying by $\bB_s\bV_A(\bV^T_A\bB_s^T\bB_s\bV_A)^{-1}\bV^T_A\bB^T_s$, yielding the fitted values \begin{equation}\hat{\bY}=\bJ_n\bY\bB_s(\bB_s^T\bB_s)^{-1}\bB^T_s+(\tilde{\bY}-\bJ_n\bY)\bB_s\bV_A(\bV^T_A\bB_s^T\bB_s\bV_A)^{-1}\bV^T_A\bB^T_s.\label{yfpc}\end{equation}
\end{enumerate}

The penalized FPC variant (Section~\ref{penfpcvt}) is also implemented via substeps (i)--(iii), with $(\bV^T_A\bB_s^T\bB_s\bV_A)^{-1}$ replaced by $[\bV^T_A(\bB_s^T\bB_s+\lambda_s\bV^T_A\bP_s)\bV_A]^{-1}$. However, tuning parameter selection proceeds differently. For the unpenalized FPC variant, we use cross-validation to choose $A$. For the penalized FPC variant, we use a large fixed $A$, say the number of components needed to explain 99\% of the variance, and use cross-validation to choose $\lambda_s$.

\section{Pointwise df derivations: Proofs of Theorems~\ref{pwdfvc}--\ref{thm2s}}\label{pwdfapp}
\subsection{Proof of Theorem~\ref{pwdfvc}}
By \eqref{vccrit} and \eqref{hvc}, the hat matrix equals
\[\mathbf{\cal H}=(\bB_s\otimes\bX)[(\bB_s^T\hat{\bSigma}^{-1}\bB_s)\otimes(\bX^T\bX)+(\bP_s\otimes\bLambda)]^{-1}[(\bB_s^T\hat{\bSigma}^{-1})\otimes\bX^T].\]
Thus, by \eqref{dfh},
\begin{eqnarray}d\el 
& = & \tr\left[\{\bb_{s}(s\el)^T\otimes\bX\}\left\{(\bB_s^T\hat{\bSigma}^{-1}\bB_s)\otimes(\bX^T\bX)+(\bP_s\otimes\bLambda)\right\}^{-1}\right. \nonumber\\
&&\qquad\qquad\left.\times\{(\bB_s^T\hat{\bSigma}^{-1}\bone_L)\otimes\bX^T\}\right].\label{dfx}\end{eqnarray}
This expression can be simplified by noting that
\begin{eqnarray*}(\bB_s^T\hat{\bSigma}^{-1}\bone_L)\otimes\bX^T&=&(\bB_s^T\hat{\bSigma}^{-1}\bB_s\bone_{K_s})\otimes[(\bX^T\bX)(\bX^T\bX)^{-1}\bX^T]\mbox{ [by \eqref{as1}]}\\
&=&\left[(\bB_s^T\hat{\bSigma}^{-1}\bB_s)\otimes(\bX^T\bX)\right]\left[\bone_{K_s}\otimes\{(\bX^T\bX)^{-1}\bX^T\}\right]\\
&=&\left[(\bB_s^T\hat{\bSigma}^{-1}\bB_s)\otimes(\bX^T\bX)+(\bP_s\otimes\bLambda)\right]\\&&\qquad\qquad\times\left[\bone_{K_s}\otimes\{(\bX^T\bX)^{-1}\bX^T\}\right]\mbox{ [by \eqref{as2}]}.
\end{eqnarray*}
 Substituting  into (\ref{dfx}) yields
\begin{eqnarray*}d\el &=& \tr\left[\{\bb_{s}(s\el)^T\otimes\bX\}\left\{\bone_{K_s}\otimes\left[(\bX^T\bX)^{-1}\bX^T\right]\right\}\right] \\
&=& [\bb_{s}(s\el)^T\bone_{K_s}] \mbox{ }\tr[\bX(\bX^T\bX)^{-1}\bX^T] \\
& = & 1\cdot p\end{eqnarray*}
for each $\ell$, where the last step uses \eqref{as1} again.

\subsection{Proof of Theorem~\ref{pp}} 
The hat matrix equals
\[\mathbf{\cal H}=\left\{\begin{array}{ll}(\bB_s \otimes \bB_t)[(\bB_s^T  \bB_s)\otimes (\bB_t^T \bB_t)+{\cal P}]^{-1}(\bB_s^T \otimes \bB_t^T), & \mbox{for penalized OLS};\\(\bB_s \otimes \bB_t)[(\bB_s^T \hat{\bSigma}^{-1} \bB_s)\otimes (\bB_t^T \bB_t)+{\cal P}]^{-1}[(\bB_s^T \hat{\bSigma}^{-1})\otimes \bB_t^T], & \mbox{for penalized GLS}.\end{array}\right.\]
By \eqref{dfh}, for $\ell=1,\ldots,L$,
\begin{eqnarray}
d\el &=& \mbox{tr}[(\be\el^T \otimes \bI_n)\mathbf{\cal H} (\bone_L \otimes \bI_n)] \nonumber\\
&=&\mbox{tr}\left[(\be\el^T \otimes \bI_n)(\bB_s\otimes \bB_t){\cal M}^T\right] \nonumber \\ 
&=&\bone_n^T[\{(\be\el^T \otimes \bI_n) (\bB_s \otimes \bB_t)\} \odot {\cal M}] \bone_{K_s K_t}. \label{deltp}
\end{eqnarray}
Since
\begin{eqnarray*}\left(\begin{array}{c}\bone_n^T[\{(\be_1^T \otimes \bI_n) (\bB_s \otimes \bB_t)\} \odot {\cal M}]\\\vdots\\\bone_n^T[\{(\be_L^T \otimes \bI_n) (\bB_s \otimes \bB_t)\} \odot {\cal M}]\end{array}\right)
&=&\left(\begin{array}{cccc}\bone_n^T&0&\ldots&0\\
0&\bone_n^T&\ldots&0\\
\vdots&&\ddots&\\
0&\ldots&0&\bone_n^T
\end{array}\right)
\\&&\qquad\times\left[\left\{\left(\begin{array}{c}\be_1^T \otimes \bI_n \\\vdots\\\be_L^T \otimes \bI_n\end{array}\right)(\bB_s \otimes \bB_t)\right\}
\odot\left(\begin{array}{c}{\cal M}\\\vdots\\{\cal M}\end{array}\right)\right] \\
&=&(\bI_L\otimes\bone_n^T)[(\bB_s \otimes \bB_t)\odot(\bone_L\otimes{\cal M})],\end{eqnarray*}
it follows that \eqref{deltp} is the $\ell$th component of \eqref{tpd}, as required.

\subsection{Proof of Theorem~\ref{chiouthm}}
Let $\hat{\bY}_1$ and $\hat{\bY}_2$ denote the two summands in \eqref{yhatchiou}.
By \eqref{dfh}, 
\[d\el=\tr\left[(\be\el^T\otimes\bI_n)\mathbf{\cal H}_1(\bone_L\otimes\bI_n)\right]+\tr\left[(\be\el^T\otimes\bI_n)\mathbf{\cal H}_2(\bone_L\otimes\bI_n)\right],\]
where $\mathbf{\cal H}_1,\mathbf{\cal H}_2$ are given by $\vecop(\hat{\bY}_k)=\mathbf{\cal H}_k\vecop(\bY)$ for $k=1,2$. Let $\bd^{(k)}=[d^{(k)}_1,\ldots,d^{(k)}_L]^T$ ($k=1,2$) denote the corresponding contributions to the pointwise df: thus $\bd=\bd^{(1)}+\bd^{(2)}$. The proof proceeds by deriving $\bd^{(1)}$ and $\bd^{(2)}$.

By \eqref{yhatchiou},  $\mathbf{\cal H}_1=[\bB_s(\bB_s^T\bB_s)^{-1}\bB_s^T] \otimes \bJ_n$. We can derive $\bd^{(1)}$ by means of the following lemma, whose proof is straightforward and is therefore omitted.  
 \begin{lemma} \label{kronlem} If $\bA$ is an $L\times L$ matrix and $\bB$ is an $n\times n$ matrix, then 
\[\left(\begin{array}{c} \tr\left[(\be_1^T\otimes\bI_n)(\bA\otimes\bB)(\bone_L\otimes\bI_n)\right]\\\vdots\\ \tr\left[(\be_L^T\otimes\bI_n)(\bA\otimes\bB)(\bone_L\otimes\bI_n)\right] \end{array}\right)=[\tr(\bB)]\bA\bone_L.\]
\end{lemma}

By Lemma~\ref{kronlem} and \eqref{as1}, 
\begin{equation}\label{dfh1}\bd^{(1)}=1\cdot\bB_s(\bB_s^T\bB_s)^{-1}\bB_s^T\bone_L=\bone_L.\end{equation}
(This is consistent with Theorem~\ref{pwdfvc}, since $\hat{\bY}_1$ represents an intercept function.)

To obtain $\mathbf{\cal H}_2$ and $\bd^{(2)}$, we require each of the linear transformations 
\[\vecop(\bY)\mymap^{(a)}\vecop(\hat{\bC})\mymap^{(b)}\vecop(\hat{\bG})\mymap^{(c)}\vecop(\hat{\bY}_2).\]
 
\begin{enumerate}[(a)]
\item By \eqref{epcs}, 
\begin{equation}\label{tr1}\vecop(\hat{\bC})=\left[\{\bV_A^T\bQ_s(\bB_s^T\bB_s)^{-1}\bB_s^T\} \otimes(\bI_n-\bJ_n)\right]\vecop(\bY).\end{equation}   
\item Similar to \eqref{ytilde}, we have $\hat{\bG}=\bA_t[\bM^*\odot(\bA_t^T\hat{\bC})]$, and hence 
\begin{eqnarray}\vecop(\hat{\bG})&=&(\bI_A\otimes\bA_t)\vecop[\bM^*\odot(\bA_t^T\hat{\bC})]\nonumber\\
&=&(\bI_A\otimes\bA_t)\bD_{M^*}(\bI_A\otimes\bA^T_t)\vecop(\hat{\bC}),
\end{eqnarray}
where $\bD_{M^*}=\mbox{Diag}[\vecop(\bM^*)]$, in view of the following lemma.
\begin{lemma} \label{nkilem} Let $\bQ,\bR,\bS$ be matrices of dimension $a\times c$, $a\times b$ and $b\times c$,
 respectively. Then 
 \[\vecop[\bQ\odot(\bR\bS)]=\bD_Q(\bI_c\otimes\bR)\vecop(\bS),\] where $\bD_{Q}=\mbox{Diag}[\vecop(\bQ)]$.
\end{lemma}
\begin{proof} Both expressions are equal to $\vecop(\bQ)\odot\vecop(\bR\bS)$.\end{proof}
\item By  \eqref{yhatchiou}, 
\begin{equation}\label{tr3}\vecop(\hat{\bY}_2) = [(\bB_s\bV_A)\otimes\bI_n]\vecop(\hat{\bG}).\end{equation}
 \end{enumerate}
 Combining \eqref{tr1}--\eqref{tr3},
\begin{eqnarray*}\mathbf{\cal H}_2&=& [(\bB_s\bV_A)\otimes\bI_n](\bI_A\otimes\bA_t)\bD_{M^*}(\bI_A\otimes\bA^T_t)\left[\{\bV_A^T\bQ_s(\bB_s^T\bB_s)^{-1}\bB_s^T\} \otimes(\bI_n-\bJ_n)\right]\\
&=& [(\bB_s\bV_A)\otimes\bA_t]\bD_{M^*}\left[\{\bV_A^T\bQ_s(\bB_s^T\bB_s)^{-1}\bB_s^T\} \otimes\bA^{cT}_t\right].
\end{eqnarray*}
Thus
\begin{eqnarray*}d\el^{(2)}&=&\tr\left[(\be\el^T\otimes\bI_n)\mathbf{\cal H}_2(\bone_L\otimes\bI_n)\right] \\ 
&=&\tr\left[ \{(\bb_s(s\el)^T\bV_A)\otimes\bA_t\}\bD_{M^*}\left(\{\bV_A^T\bQ_s(\bB_s^T\bB_s)^{-1}\bB_s^T\bone_L\} \otimes\bA^{cT}_t\right) \right]\\ 
&=&\tr\left[ \{(\bb_s(s\el)^T\bV_A)\otimes\bA_t\}\bD_{M^*}\{(\bV_A^T\bQ_s\bone_{K_s}) \otimes\bA^{cT}_t\}\right],
\end{eqnarray*}
by \eqref{as1}. By Theorem~8.15(b) of \cite{schott2005}, this implies
\begin{eqnarray}d\el^{(2)}&=&\bone_n^T\left[ \{(\bb_s(s\el)^T\bV_A)\otimes\bA_t\}\odot\{(\bone_{K_s}^T\bQ_s\bV_A) \otimes\bA^{c}_t\} \right]\vecop(\bM^*)\nonumber\\
&=&\bone_n^T\left[\{(\bb_s(s\el)^T\bV_A)\odot(\bone_{K_s}^T\bQ_s\bV_A)\}\otimes(\bA_t\odot\bA^{c}_t) \right]\vecop(\bM^*) \nonumber\\
&=&\bone_n^T(\bA_t\odot\bA^{c}_t)\bM^*\left[\{\bV_A^T\bb_s(s\el)\}\odot(\bV_A^T\bQ_s\bone_{K_s})\right].\label{dfh2}
\end{eqnarray}
Thus $\bd^{(2)}$ equals the second term of \eqref{dfchiou}.
Combining this with \eqref{dfh1} completes the proof.

\subsection{Proof of Theorem~\ref{thm2s}}
\paragraph{Parts (a) and (b).}
By (\ref{ytilde}) and Lemma~\ref{nkilem},
$\tilde{\by}=(\bI_L\otimes\bA_t)\bD_M(\bI_L\otimes\bA_t^T)\by$,
where $\bD_M=\mbox{Diag}[\vecop(\bM)]$.  Thus, by \eqref{fit2}, 
\[\hat{\by} =(\bH_s\otimes\bI_n)\tilde{\by}= \mathbf{\cal H}\by\]
 with $\mathbf{\cal H}=(\bH_s\otimes\bA_t)\bD_M(\bI_L\otimes\bA_t^T)$.  Substituting this into (\ref{dfh}) yields, for $\ell=1,\ldots,L$, 
\[d\el=\tr\left[\{(\be\el^T\bH_s)\otimes\bA_t\}\bD_M(\bone_L\otimes\bA_t^T)\right].\]
 Mimicking the steps leading to \eqref{dfh2}, we obtain
 $d\el=\bone_n^T(\bA_t\odot\bA_t)\bM\bH_s^T\be\el$, 
  and thus \[\bd=\bH_s\bM^T(\bA_t^T\odot\bA_t^T)\bone_n.\] Replacing $\bH_s$ by $\bI_L$ in the above leads to 
\begin{equation}\label{tbd}\tilde{\bd}=\bM^T(\bA_t^T\odot\bA_t^T)\bone_n\end{equation}
and hence $\bd=\bH_s\tilde{\bd}$, proving part (\ref{dfpen}). 

Part (\ref{step1df}) follows from \eqref{tbd} implies upon noting that 
\begin{equation}\label{aon}\bA_t^T\bA_t=\bI_{K_t}\end{equation} and hence $(\bA_t^T\odot\bA_t^T)\bone_n=\bone_{K_t}$.
Alternatively, part~(\ref{step1df}) can be proved directly by noting that $\tilde{\bd}=\bM^T\bone_{K_t}$ is equivalent to $\tilde{d}\el=\sum_{i=1}^{K_t}\frac{1}{1+\lambda_{x \ell} \tau_i}$ for $\ell=1,\ldots,L$. This is a standard df formula for Demmler-Reinsch orthogonalization \citep[][p.\ 336]{ruppert2003}.

\paragraph{Part (\ref{dffpc}).}
By \eqref{yfpc},  $\bd=\bd^{(1)}+\bd^{(2)}+\bd^{(3)}$ where $\bd^{(1)},\bd^{(2)},\bd^{(3)}$ are the respective contributions of 
$\bJ_n\bY\bB_s(\bB_s^T\bB_s)^{-1}\bB^T_s$, $\tilde{\bY}\bB_s\bV_A\bN^{-1}\bV^T_A\bB^T_s$ and  $-\bJ_n\bY\bB_s\bV_A\bN^{-1}\bV^T_A\bB^T_s$.
As in the proof of Theorem~\ref{chiouthm}, 
$\bd^{(1)}=\bone_L$. Arguing as in the proof of Theorem~\ref{thm2s}(b) but with $\bB_s\bV_A\bN^{-1}\bV^T_A\bB^T_s$ replacing $\bH_s$, $\bd^{(2)}=\bB_s\bV_A\bN^{-1}\bV^T_A\bB^T_s\tilde{\bd}$. Invoking Lemma~\ref{kronlem} as in the proof of Theorem~\ref{chiouthm} gives 
$\bd^{(3)} =-\bB_s\bV_A\bN^{-1}\bV^T_A\bB^T_s\bone_L$. Combining these results yields \eqref{dfppp}.

\section{Functional $R^2$ for simulated data}\label{r2app}
We noted in Section~\ref{simdesign} that the simulated functional responses were given by $y_i(s)=f(t_i,s)+\varepsilon_i(s)$, where $\varepsilon_i(s)$ was generated via the sum of two independent processes with variances $\gamma\sigma^2$ and $\sigma^2$. We further noted that $\sigma^2$ was chosen to attain specified values of $R^2$. We now show how this is done.

Suppose we have a preliminary set of responses 
\begin{equation}\label{yprelim}y_i^*(\cdot)=f(t_i,\cdot)+\varepsilon^*_i(\cdot),\end{equation} $i=1,\ldots,n$, 
with the $\varepsilon^*_i(\cdot)$'s given by \eqref{vee}, as described in Section~\ref{simdesign}, for some $\sigma^2,\gamma$; let $\bar{y}^*(\cdot)$ denote their sample mean. By \eqref{r2def} and the identity $y^*_i(s)-\bar{y}^*(s)=\varepsilon^*_i(s)+[f(t_i,s)-\bar{y}^*(s)]$, the true-model coefficient of determination for these preliminary responses is $R^2(y^*_1,\ldots,y^*_n)=1-\frac{A}{A+B+2C}$ where
\begin{eqnarray*}A&=&\sum_{i=1}^n \sum_{s\in S}\varepsilon^*_i(s)^2,\\
B&=&\sum_{i=1}^n \sum_{s\in S}[f(t_i,s)-\bar{y}^*(s)]^2,\\
C&=&\sum_{i=1}^n \sum_{s\in S}\varepsilon^*_i(s)[f(t_i,s)-\bar{y}^*(s)].\end{eqnarray*}
 If we define 
 \begin{equation}\label{ysig}\varepsilon^{(\kappa)}_i(s)=\kappa\varepsilon^*_i(s)\mbox{ and }y^{(\kappa)}_i(s)=f(t_i,s)+\varepsilon^{(\kappa)}_i(s)\end{equation} 
 $(i=1,\ldots,n)$ for any $\kappa>0$, then 
 $R^2(y^{(\kappa)}_1,\ldots,y^{(\kappa)}_n)\approx 1-\frac{\kappa^2A}{\kappa^2A+B+2\kappa C}$. Exact equality does not hold here because the mean of $y^{(\kappa)}_1,\ldots,y^{(\kappa)}_n$ differs slightly from $\bar{y}^*$. However, this approximation serves as the basis for the following iterative algorithm.
 \begin{enumerate}
 \item Fix $\sigma^2=1$, choose some $\gamma>0$, and generate the preliminary responses \eqref{yprelim} as above.
 \item\label{stepp2} Obtain modified responses \eqref{ysig}, with $\kappa$ chosen so that $R^2=1-\frac{\kappa^2A}{\kappa^2A+B+2\kappa C}$ for the desired $R^2$. By the quadratic formula, we can take $\kappa=\frac{C(1-R^2)+\sqrt{C^2(1-R^2)^2+ABR^2(1-R^2)}}{AR^2}$.
 \item Compute the actual $R^2$ \eqref{r2def} for  $y^{(\kappa)}_1,\ldots,y^{(\kappa)}_n$. If it is not within a set tolerance (say, 0.0001) of the desired $R^2$, set $\varepsilon^*_i=\varepsilon^{(\kappa)}_i$ for each $i$ and return to step~\ref{stepp2}; otherwise set $y_i=y^{(\kappa)}_i$ for each $i$, and the algorithm is done.
\end{enumerate} 
In practice we have found this algorithm to converge very quickly. Note that if $\kappa_1,\ldots,\kappa_m$ are the values of $\kappa$ derived in step~\ref{stepp2} of the successive iterations, then the final responses $y_1,\ldots,y_n$ have in effect been generated with variance parameter $\sigma^2=\kappa_1^2\ldots \kappa_m^2$.

\section{Interval estimation for the penalized variant of the two-step method}\label{intapp}
By (\ref{gest}) and Lemma~\ref{nkilem}, the two-step estimate of $f(t,s)$, with the penalized variant  of step~2 (Section~\ref{penvt}), is
\begin{equation}\hat{f}(t,s) = \left[\{\bb_s(s)^T(\bB_s^T\bB_s+\lambda_s\bP_s)^{-1}\bB_s^T\}\otimes \{\bb_t(t)^T\bR_t^{-1}\bU_t\}\right]\bD_M(\bI_L\otimes\bA_t^T)\by.\label{gxs}
\end{equation}
 Assuming between-curve independence, we have 
 $\widehat{\var}(\by)=\hat{\bSigma}\otimes\bI_n$ 
 where $\hat{\bSigma}$ is an estimate of $\left(\cov[y(s_i),y(s_j)|x]\right)_{1\leq i,j\leq L}$. 
Combining this with (\ref{gxs}) yields 
\[\widehat{\var}[\hat{f}(t,s)]=\|\bT_0[\bb_s(s)\otimes\bb_t(t)]\|^2\]
where 
$\bT_0=(\hat{\bSigma}^{1/2}\otimes\bA_t)\bD_M\left[\{\bB_s(\bB_s^T\bB_s+\lambda_s\bP_s)^{-1}\}\otimes \{\bU_t^T\bR_t^{-T}\}\right]$.
 By (\ref{aon}), we can use the equivalent but more computationally efficient formula
\begin{equation}\label{nvec}\widehat{\var}[\hat{f}(t,s)]=\|\bT[\bb_s(s)\otimes\bb_t(t)]\|^2\end{equation}
where 
 $\bT=(\hat{\bSigma}^{1/2}\otimes\bI_{K_t})\bD_M\left[\{\bB_s(\bB_s^T\bB_s+\lambda_s\bP_s)^{-1}\}\otimes \{\bU_t^T\bR_t^{-T}\}\right]$.
 
 Given the $LK_t\times K_sK_t$ matrix $\bT$, it is straightforward to compute a matrix of pointwise variance estimates $\hat{\bV}=\left(\widehat{\var}[\hat{f}(t_g^*,s_h^*)]\right)_{1\leq g\leq G, 1\leq h\leq H}$, since \eqref{nvec} implies 
\begin{equation}\label{lasteq}\vecop(\hat{\bV})= \bone_{LK_t}^T\left[\left\{\bT(\bB_s^{*T}\otimes\bB_t^{*T})\right\}^{\odot 2}\right],\end{equation}
where 
$\bB_s^*=[b_{s.j}(s^*_h)]_{1\leq h\leq H,1\leq j\leq K_s}$, $\bB_t^*=[b_{t.j}(t^*_g)]_{1\leq i\leq G,1\leq j\leq K_t}$ and $\bE^{\odot 2}\equiv\bE\odot\bE$. Pointwise confidence intervals based   on \eqref{lasteq} treat the smoothing parameters as fixed. This shortcut of ignoring the variability due to smoothing parameter selection is fairly standard for ordinary semiparametric regression, but further study is required to assess its impact on coverage for two-step varying-smoother models.

\bibliographystyle{chicago}
\bibliography{vsm-arxiv}

\end{document}